\documentclass[11pt,a4paper]{article}
\usepackage[T1]{fontenc}
\usepackage[utf8]{inputenc}

\newcommand{\pdfextension}{pdf}
\newcommand{\pngextension}{png}
\newcommand{\figuredir}{figures/}
\newcommand{\hdimages}{_hd} 
\def\myqed{} 
\def\mysubsection#1{\subsection{#1}} 

\pdfoutput=1 

\usepackage{amsmath}
\usepackage{amsfonts}
\usepackage{graphicx}
\usepackage{times}
\usepackage{amsthm}
\usepackage{amssymb}
\usepackage{color}
\usepackage{mhequ}
\usepackage{dsfont}
\usepackage[margin=2.7cm]{geometry}
\usepackage{color}
\usepackage{url}
\usepackage{authblk}
\usepackage{subfig}
\usepackage{mhequ}
\usepackage{cite}

\usepackage[expansion=true]{microtype}
\usepackage[plainpages=false]{hyperref}

\hypersetup{
    colorlinks=true,
    linkcolor=black,
    citecolor=black,
    filecolor=black,
    urlcolor=black,
}

\newcommand{\vertiii}[1]{{\left\vert\kern-0.25ex\left\vert\kern-0.25ex\left\vert #1 
    \right\vert\kern-0.25ex\right\vert\kern-0.25ex\right\vert}}

\let\ssref=\ref
\def\fref#1{Figure~\ssref{#1}}

\def\cref#1{Condition~\ssref{#1}}
\def\Cref#1{Corollary~\ssref{#1}}
\def\eref#1{(\ssref{#1})}
\def\sref#1{\textsection\ssref{#1}}
\def\lref#1{Lemma~\ssref{#1}}
\def\rref#1{Remark~\ssref{#1}}
\def\tref#1{Theorem~\ssref{#1}}
\def\dref#1{Definition~\ssref{#1}}
\def\pref#1{Proposition~\ssref{#1}}
\def\aref#1{Assumption~\ssref{#1}}

\def\ref#1{\defcommand#1?%
	\myref{#1}}

\def\Label#1{}

\def\thecomma{\ifx,\thenext \else\ifx;\thenext \else\ifx.\thenext
\else\ifx!\thenext \else\ifx:\thenext\else\ifx)\thenext \else \
\fi\fi\fi\fi\fi\fi}
\def\condblank{\futurelet\thenext\thecomma}
\def\ie{{\it i.e.,}\condblank}
\def\eg{{\it e.g.,}\condblank}

\newenvironment{myenum}
{\begin{enumerate}
  \setlength{\itemsep}{1pt}
  \setlength{\parskip}{0pt}
  \setlength{\parsep}{0pt}}
{\end{enumerate}}

\numberwithin{equation}{section}

\newtheorem{theorem}{Theorem}[section]
\newtheorem{lemma}[theorem]{Lemma}
\newtheorem{proposition}[theorem]{Proposition}
\newtheorem{definition}[theorem]{Definition}

\newtheorem{assumption}[theorem]{Assumption}

\theoremstyle{definition} 
\newtheorem{remark}[theorem]{Remark}

\let\rho=\varrho
\let\kappa=\varkappa
\let\phi=\varphi

\def\CC{{\mathcal C}}

\def\OO{{\mathcal O}}
\def\SS{{\mathcal S}}
\def\BB{{\mathcal B}}

\newcommand{\dd}{\mathrm{d}}
\renewcommand{\d}{\,\mathrm{d}}
\newcommand{\dt}{\,\dd t}

\newcommand{\avg}[1]{\left\langle #1\right\rangle}

\def\argcdot{{\,\cdot\,}}

\def\lr{\Leftrightarrow}
\newcommand{\ind}{\mathbf{1}}
\newcommand{\pti}[1]{\tilde p_{#1}}
\newcommand{\ptipow}[2]{\tilde p_{#1}^{\,#2}}

\def\CstOmegaia{{C_1}}
\def\CstOmegaib{{C_2}}
\def\CstOmegaid{{C_3}}

\def\CstOmegaca{{C_4}}
\def\CstOmegacb{{C_5}}
\def\CstOmegacd{{C_{6}}}

\def\CstEbha{{C_7}}
\def\CstEbhb{{C_8}}

\def\CstGa{{C_{9}}}

\def\pworkpower{r}
\def\pworkpowerd{r}

\def\tH{{\widetilde H}}

\def\L{{\mathrm L}}
\def\R{{\mathrm R}}
\def\C{{\mathrm C}}
\def\I{{\mathrm I}}
\def\Wl#1{{W_\L^{[#1]}}}
\def\Wr#1{{W_\R^{[#1]}}}
\def\Wc#1{{W_\C^{[#1]}}}

\def\Hc#1{{H_c^{(#1)}}}
\def\barL{\bar{L}}
\def\barR{\bar{K}}
\def\bEE{{\bar{\mathbb E}}}
\let\epsilon=\varepsilon
\def\OA{A_1}
\def\OB{A_2}
\def\OC{A_3}

\def\Pnull{{\mathcal P}}
\def\torus{{\mathbb T}}
\def\real{{\mathbb R}}
\def\RAB{R_{12}}
\def\RBC{R_{23}}
\def\Y{M}
\def\barLO{\bar{L}_e}
\def\Enull{{\bar{\mathbb E}}^e}

\begin{document}

\title{Non-equilibrium steady states for chains of four rotors}
\author[1]{N. Cuneo}
\author[1,2]{J.-P. Eckmann}
\affil[1]{D\'epartement de physique th\'eorique, Universit\'e de Gen\`eve}
\affil[2]{Section de math\'ematiques, Universit\'e de Gen\`eve}
\date{} 

\maketitle
\thispagestyle{empty}

\begin{abstract}
We study a chain of four interacting rotors (rotators)
connected at both ends to stochastic heat baths at different temperatures.
We show that for non-degenerate interaction potentials
the system relaxes, at a stretched exponential rate, to a
non-equilibrium steady state (NESS).  
Rotors with high energy tend to decouple from their neighbors due to 
fast oscillation of the forces. Because of this,
the energy of the central two rotors, which interact with the heat baths
only through the external rotors, can take a very long time to dissipate.
By appropriately averaging the oscillatory forces, we estimate
the dissipation rate and construct a Lyapunov function.
Compared to the chain of length three (considered previously by C.~Poquet
and the current authors), the new difficulty with four rotors is
the appearance of resonances when both central rotors
are fast. We deal with these resonances using the rapid thermalization
of the two external rotors.
\end{abstract}

\tableofcontents

\section{Introduction}
We consider a chain of 4 classical rotors interacting at both ends with
stochastic heat baths at different temperatures. 
Under the action of such heat baths, many Hamiltonian systems
are known to relax to an invariant probability measure called non-equilibrium
steady state (NESS).
In general, the explicit expression for this invariant measure is
unknown, and the convergence rate depends on the nature of the system.
For the model under consideration, we obtain a stretched exponential
rate.

For several examples of Hamiltonian chains, properties of the NESS
(\eg  thermal conductivity,
validity of the Fourier law, temperature profile, \dots)
have been studied numerically, perturbatively,
or via some effective theories. See for example
\cite{lepri_thermal_2003, iacobucci_negative_2011, bonetto_fouriers_2000, gendelman_normal_2000, roeck_asymptotic_2014,MR2148380,bernardin2014green} for chains of
rotors and \cite{lepri_thermal_2003,bonetto_fouriers_2000, aoki_energy_2006, bricmont_towards_2007, giardina_finite_2000,lefevere_normal_2006,bernardin2014green} for chains of oscillators.
From a rigorous point of view however, the mere existence of
an invariant measure is
not evident, and has been proved only in special cases.

A lot of attention has been devoted to chains of classical oscillators
with (nonlinear) nearest neighbor interactions.
In such models, each oscillator has
a position $q_i\in \real$ (we take one dimension for simplicity), is attached
to the reference frame with a {\em pinning} potential $U(q_i)$, and 
interacts with its neighbors via some {\em interaction} potentials
$W(q_{i+1}-q_i)$ and $W(q_{i}-q_{i-1})$.

It turns out that the properties of the chain depend crucially
on the relative growth of $W$ and $U$ at high energy.
In the case of (asymptotically) polynomial potentials,
and for Markovian heat baths, it has been shown
\cite{eckmann_hairer_2000,eckmann_entropy_1998,eckmann_nonequilibrium_1999,
eckmann_temperature_2004,reybellet_exponential_2002,carmona_2007}
that if $W$ grows faster than $U$, the system typically relaxes exponentially
fast to a NESS. The convergence is fast because, thanks to the strong
interactions,
the sites in the bulk of the chain ``feel'' the heat baths effectively
even though they are separated from them by other sites.

In the strongly pinned case, \ie when $U$ grows faster than $W$,
the situation is more complicated. When a given
site has a lot of energy, the corresponding oscillator
essentially feels only its pinning potential $U(q_i)$ and not the
interaction. Assume $U(q) \propto q^{2k}$ with $k>1$.
An isolated oscillator pinned with a potential $U$ and with an energy $E$
oscillates with a frequency that grows like $E^{1/2 - 1/{2k}}$.
This scaling plays a central role, since the larger the energy at a site,
the faster the corresponding $q_i$ oscillates. But then, the interaction forces
with the sites $i+1$ and $i-1$ oscillate very rapidly and become ineffective
at high energy. Therefore, a site (or a set of sites) with high energy tends to
decouple
from the rest of the chain, so that energy can be ``trapped'' in the bulk.
This mechanism not only makes the convergence to the invariant measure slower, 
but it also makes the proof of its existence harder.
The case where $W$ is quadratic is considered in \cite{hairer_slow_2009}.
There, Hairer and Mattingly show that if $U(q) \propto q^{2k}$
with $k$ sufficiently large, no exponential convergence
to an invariant measure (if there is one) can take place. Moreover, they show
that an invariant measure exists in the case of $3$ oscillators when $k>3/2$.
The existence of a NESS for longer chains of
oscillators remains an open problem when the pinning dominates the
interactions.

Chains of rotors are in fact closely
related to strongly pinned oscillator chains:
The frequency of a rotor scales as $E^{1/2}$, where $E$ is its energy.
This scaling corresponds to that of an oscillator in the limit $k\to \infty$,
for the pinning $U(q) \propto q^{2k}$ discussed above. In this sense, 
our chain of rotors behaves as a chain of oscillators in the limit of
``infinite pinning'', which is some kind of worst-case scenario
from the point of view of the asymptotic decoupling at high energy. 
On the other hand, the compactness of the position-space  (it is a torus)
in the rotor case is technically very convenient.
The problems appearing with chains of strongly pinned oscillators
are very similar to those faced with chains of rotors, and so are
the ideas involved to solve them.

The existence of an invariant measure
for the chain of 3 rotors has been proved in \cite{cuneo_nonequilibrium_2014},
as well as a stretched exponential upper bound of the kind $\exp({-c\sqrt{t}})$
on the convergence rate.
The methods, which involve averaging the rapid oscillations
of the central rotor, are inspired by
those of \cite{hairer_slow_2009} for the chain of 3 oscillators.

In the present paper, we generalize the result of
\cite{cuneo_nonequilibrium_2014}
to the case of 4 rotors, and obtain again a bound $\exp({-c\sqrt{t}})$
on the convergence rate.
The main new difficulty in this generalization
is the presence of resonances among the two central rotors.
When they both have a large energy, there are two
fast variables and some resonant terms
make the averaging technique developed in \cite{cuneo_nonequilibrium_2014}
insufficient. A large portion of the present paper is devoted
to eliminating such resonant terms by using the rapid thermalization
of the external rotors.

It would be of course desirable to be able to work with a larger
number of rotors. The present paper uses explicit methods to deal
with the averaging phenomena. We hope that by crystallizing the
essentials of our methods, longer chains can be handled in the same
spirit. We expect that for longer chains, the convergence rate
is of the form $\exp(-ct^k)$, for some exponent $k\in (0,1)$ which depends
on the length of the chain. We formulate a conjecture
and explain the main difficulties for longer chains
in \rref{r:conjecture}.

We now introduce the model and state the main results. 
In \sref{sec:23fast}, we study the behavior of the system when one
of the two central rotors is fast, and construct a Lyapunov function
in this region. In \sref{sec:bothfast}, we do the same in the regime
where both central rotors are fast. In \sref{s:constrlyapunov}
we construct a Lyapunov function that is valid across all regimes, and in
\sref{s:controltechnical} we provide the technicalities necessary to
obtain the main result.

\mysubsection{The model}
\begin{figure}[ht]
\centering
\includegraphics[width=3.8in]{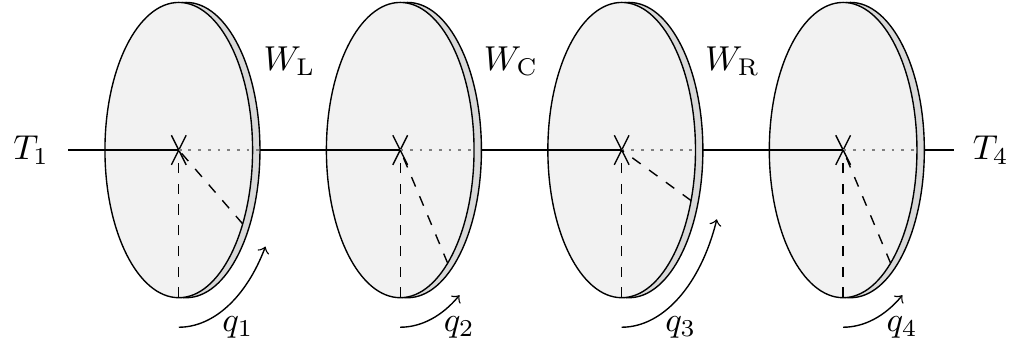}
\caption{A chain of four rotors with two heat baths at temperatures $T_1$ and
$T_4$.}
\label{f:system}
\end{figure}

We study a model of 4 rotors, each given by a momentum $p_i \in \real$
and an angle $q_i \in \torus = \real / 2\pi \mathbb Z$, $i=1, \dots, 4$. 
We write in the sequel $q= (q_1, \dots, q_4) \in \torus^4$,
$p = (p_1, \dots, p_4) \in \real^4$, and $x = (q, p) \in \Omega$,
where $\Omega = \torus^4 \times \real^4$ is the phase space of the
system. We consider the Hamiltonian 
\begin{equ}[eq:hamiltonien]
H(x) = \sum_{i=1}^4 \frac{p_i^2}{ 2}  + W_\L(q_2-q_1) + W_{\C}(q_3-q_2) +
W_{\R}(q_3-q_4)~,
\end{equ}
where $W_\I : \torus \to \real$, $\I = \L, \C, \R$ (standing for left, center
and right) are smooth $2\pi$-periodic interaction potentials (see \fref{f:system}).

\medskip
\noindent{\bf Convention:} Unless specified otherwise,
the arguments of the potentials are always as
above, namely $W_\L= W_\L(q_2-q_1)$, $W_{\C}= W_{\C}(q_3-q_2)$ and
$W_{\R}= W_{\R}(q_3-q_4)$. The same applies to any function
with index $\L, \C$ and $\R$.
Note that the argument for $\R$ is $q_3-q_4$ (and not
$q_4-q_3$) since this choice will lead to more symmetrical
expressions between the sites 1 and 4.
\medskip

To model the interaction with two heat baths, we add at each end of the chain
a Langevin thermostat at temperature $T_b> 0$, with dissipation
constant $\gamma_b > 0$, $b=1,4$.
Introducing the derivative of the potentials $w_\I = W_\I'$, $\I = \L,\C,\R$,
the main object of our study is the SDE:
\begin{equa}[e:SDE]
\dd q_i &= p_i\dt ~,\qquad i=1,\dots,4~,\\
\dd p_1 &=  ( w_\L - \gamma_1 p_1 )\dt + \sqrt{2\gamma_1 T_1} \dd B^1_t~,\\
\dd p_2 &=  (w_{\C}- w_\L)\dt~, \\
\dd p_3 &=  - (w_{\C} + w_{\R})\dt~, \\
\dd p_4 &=  (w_{\R} - \gamma_4 p_4 )\dt+ \sqrt{2\gamma_4 T_4} \dd B^4_t~,
\end{equa}
where $B^1_t, B^4_t$ are independent standard Brownian motions. 
The generator of
the semigroup associated to \eref{e:SDE} reads
\begin{equs}[e:genl]
L &= \sum_{i =1}^4 p_i\partial_{q_i} +
w_\L\cdot(\partial_{p_1}-\partial_{p_2})+w_\C\cdot(\partial_{p_2}-\partial_{p_3})
+w_\R\cdot(\partial_{p_4}-\partial_{p_3})\\
& \qquad \qquad +\sum_{b=1,4}\left(-\gamma_b p_b
\partial_{p_b} + \gamma_bT_b \partial^2_{p_b}\right)~.
\end{equs}	

\begin{remark}In contrast to \cite{cuneo_nonequilibrium_2014},
we do not allow for the presence of pinning potentials $U(q_i)$ and of
constant forces at the ends of the chain, although we believe that the
main result still holds with such modifications. While constant
forces would be easy to handle, the addition of a pinning potential
would require some generalization of a technical result (\pref{prop:rbarref})
which we are currently unable to provide (see \rref{rem:paspinningettau}).
\end{remark}

We consider the measure space $(\Omega, \mathcal B)$, with the Borel
$\sigma$-field $\mathcal B$ over $\Omega$.
The coefficients in \eref{e:SDE} are globally Lipschitz,
and therefore the solutions are almost surely defined for all
times and all initial conditions.
We denote the transition probability  of the corresponding Markov process
by $P^t(x, \argcdot)$, for all $x\in \Omega$ and $t\geq 0$.

\mysubsection{Main results}

We will often refer to the sites 1 and 4 as the {\em outer} (or {\em external})
rotors, and the sites 2 and 3 as the {\em central} rotors.
We require the interactions from
the inner rotors to the outer rotors to be non-degenerate
in the following sense:

\begin{assumption}\label{as:assumptioncoupling}
We assume that for $\I = \L, \R$ and for each $s\in \torus$, at
least one of the derivatives $w^{(k)}_\I(s)$, $k\geq 1$
is non-zero.
\end{assumption}

This assumption is not very restrictive. In particular, it holds if all the
potentials consist of finitely many nonconstant Fourier modes.

Our main result is a statement about the speed of convergence
to a unique stationary state of
the system \eref{e:SDE}. In order to state it, we introduce
for each continuous function $f:\Omega \to (0, \infty)$ the norm
$\|\argcdot \|_{f}$ on the space of signed Borel
measures on $\Omega$:
$$
\|\nu \|_{f} = \sup_{|g|\leq f} \int_\Omega g \dd \nu~.
$$
If $f\equiv 1$, we retrieve the total variation norm.
\begin{theorem}[\bf NESS and rate of convergence]\label{thm:mainthm}Under
\aref{as:assumptioncoupling}, the Markov process defined by \eref{e:SDE}
satisfies:
\begin{myenum}
	\item[(i)] The transition probabilities $P^t(x, \dd y)$ have a
	$\CC ^\infty((0, \infty)\times \Omega\times \Omega)$ density
	$p_t(x, y)$.
	\item[(ii)] There is a unique invariant measure $\pi$, and it has a
	smooth density.
	\item[(iii)]For all $0 \leq \theta_1 < \min(1/T_1, 1/T_4)$ and all $\theta_2 >
\theta_1$,
	there exist constants $C, \lambda > 0$ such that for all $x = (q_1, q_2,
\dots,
p_4)\in\Omega$ and all $t\geq 0$,
	\begin{equ}[eq:stretchedexpconv]
		\|P^t(x, \argcdot) - \pi \|_{e^{\theta_1H}} \leq
Ce^{\theta_2 H(x)}e^{-\lambda t^{1/2}}~.
	\end{equ}
\end{myenum}
\end{theorem}

At thermal equilibrium, namely when $T_1 = T_4 = 1/\beta$ for some $\beta>0$,
the invariant measure is the Gibbs distribution with density $e^{-\beta
H(x)}/Z$,
where $Z$ is a normalization constant.

\tref{thm:mainthm} will be proved in \sref{s:controltechnical} with help
of results of \cite{douc_subgeometric_2009} and the existence of a 
Lyapunov function,
the properties of which are stated in
\begin{theorem}[\bf Lyapunov function]\label{prop:Lyapunov}Let
$0<\theta < \min(1/T_1, 1/T_4)$. Under \aref{as:assumptioncoupling},
there is a function $V: \Omega \to [1,\infty)$ satisfying:
\begin{myenum}
\item[(i)] There are constants $c_1>0$ and $a\in(0,1)$ such that
\begin{equ}\label{e:ineqvn}
	 1+e^{\theta H}\leq V\leq c_1(e^{|p_2|^a}+e^{|p_3|^a})e^{\theta H}~.
\end{equ}
\item[(ii)] There are a compact set $K$ and constants $c_2, c_3 > 0$
such that
\begin{equ}\label{e:ineqLvn}
	LV \leq c_2 \ind_K - \phi(V) ~,
\end{equ}
with  $\phi: [1, \infty) \to (0, \infty)$ the increasing, concave function 
 defined by
\begin{equ}\label{e:defphi}
\phi(s) = \frac {c_3\,s}{2+\log(s)}~.
\end{equ}
\end{myenum}
\end{theorem}

Most of the paper will be devoted to proving the existence
of such a Lyapunov function.

\begin{remark}\label{rem:Teq0debut} We assume throughout that $T_1$ and $T_4$
are
strictly positive. While the conclusions of \tref{prop:Lyapunov}
remain true for $T_1 = T_4 = 0$ (with any $\theta > 0$),
part of the argument has to be changed  in this case, 
as sketched in \rref{rem:Teq0Harris}. The positivity of the temperatures
is, however, essential for \tref{thm:mainthm}; at zero temperature,
the system is not irreducible, and none of the conclusions of
\tref{thm:mainthm}
hold. 
\end{remark}

\mysubsection{Overview of the dynamics}

To gain some insight into the strategy of the proof,
we illustrate some essential features
of the dynamics \eref{e:SDE}. 
Since the exterior rotors (at sites 1 and 4) are directly damped
by the $-\gamma_b p_b \dt$ terms in \eref{e:SDE}, we expect their
energy to decrease rapidly with large probability.
More specifically, for $b=1,4$, we find that $Lp_b$ is
equal to $-\gamma_b p_b$ plus some bounded terms, and
thus we expect $p_b$ to decay exponentially (in expectation value) 
when it is large. Therefore, the external rotors
recover very fast from thermal fluctuations, and will not be
hard to deal with.

On the other hand, the central rotors are not damped directly,
and feel the dissipative terms of \eref{e:SDE} only
indirectly, by interacting with the outer rotors.
The interesting issue appears when the energy of the system is very large
and mostly concentrated
in one or both of the central rotors. If most of the energy is at site
2 (meaning that $|p_2|$ is much larger than all other momenta),
the corresponding rotor spins very rapidly, \ie $q_2$ moves very rapidly
on $\torus$. But then, 
the interaction forces $w_\L(q_2-q_1)$ and $w_\C(q_3-q_2)$ oscillate
rapidly,
which causes the site to essentially decouple from its neighbors.
The same happens when most of the energy is at site 3, when  $w_\C(q_3-q_2)$ 
and $w_\R(q_3-q_4)$ oscillate rapidly. And when both $|p_2|$
and $|p_3|$ are large and much larger than $|p_1|$ and $|p_4|$, the 
forces $w_\L(q_2-q_1)$ and $w_\R(q_3-q_4)$ are highly oscillatory, so that
the central two rotors almost decouple from the outer ones (the force
$w_\C$ might or might not oscillate depending on $p_2$ and $p_3$).

This asymptotic decoupling is the interesting feature of the model:
in principle, if the central rotors do
not recover sufficiently fast from thermal fluctuations, the energy
of the chain could grow (in expectation value) without bounds. 
On the other hand, when their energy is large, the decoupling
phenomenon should make the central rotors less affected
by the fluctuations of the heat baths. Our results imply that both
effects combine in a way that prevents overheating. See \cite[Remark 3.10]{cuneo_nonequilibrium_2014}
for a quantitative discussion of these two effects for a chain
of three rotors. See also \cite{hairer_how_2009} for a clear exposition
of the overheating problem in a related model.

\begin{figure}[ht]
\centering
\includegraphics[width=0.90\textwidth, trim = 0 2mm -10mm 0, clip]{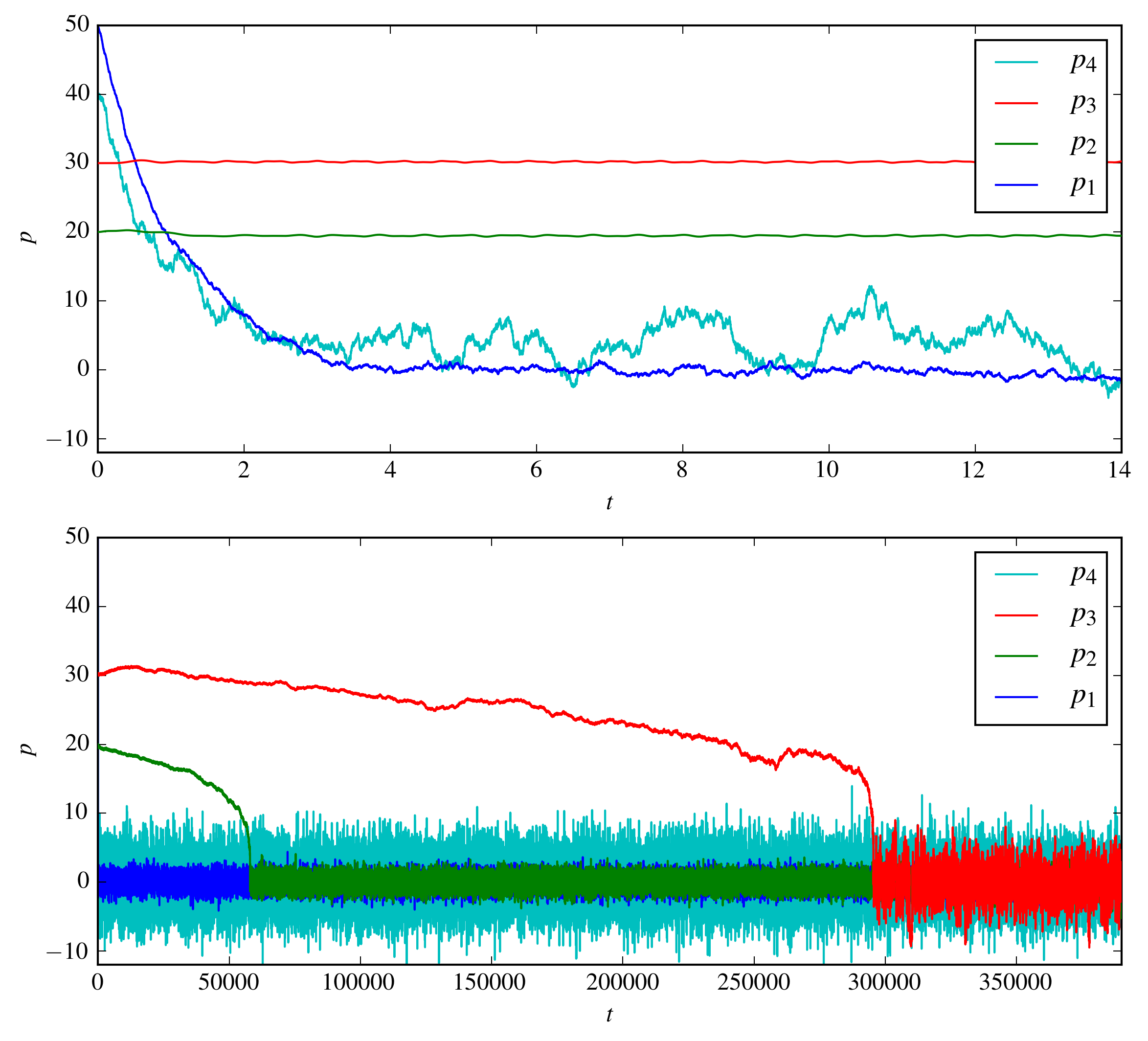}
\caption{Evolution of the momenta $p_1,\dots,p_4$, for $\gamma_1=\gamma_4 = 1$,
$T_1 = 1$, $T_4 = 10$, $q(0) = (0, 0, 0, 0)$
and $p(0) = (50, 20, 30, 40)$. The interaction potentials in the 
simulations here are $W_\I = -\cos$
so that the forces are $w_\I = \sin$, $\I=\L, \C, \R$.
}\label{fig:traj}
\end{figure}

\fref{fig:traj} illustrates the evolution\footnote{
The numerical algorithm used in this paper is based on
the one described in \cite{iacobucci_negative_2011}. The time step
is either $10^{-2}$ or $10^{-3}$ depending on the situation.}
of the momenta at two different 
time scales, starting with $p(0) = (50, 20, 30, 40)$.
The upper graph shows that indeed $p_1$ and $p_4$ decrease very fast,
and the lower graph indicates that $p_2$ and $p_3$
remain large for a significantly longer time, but eventually
also decrease. Since for this initial condition 
$p_3$ is larger than $p_2$, the force $w_\R$ oscillates
faster than $w_\L$. Therefore, $p_3$ couples less effectively to the
outer rotors (where the dissipation happens) than $p_2$, and hence $p_3$
decreases more slowly.

If one were to look at these trajectories for much
longer times, one would eventually observe some
fluctuations of arbitrary magnitude, followed by new
recovery phases. But large fluctuations are very rare.

Since the system is rapidly driven to small $p_1$, $p_4$, it is really
the dynamics of $(p_2, p_3)$ that plays the most important role.
We will often argue in terms of the 8-dimensional dynamics projected
onto the $p_2p_3$-plane. 
We illustrate some trajectories in this plane for several initial
conditions in \fref{fig:traj_avg}. To make the illustration
readable, we used a very small temperature, so that the picture is
dominated by the deterministic dynamics.

\begin{figure}[ht]
\centering
\includegraphics[width=0.95\linewidth, trim = 2mm 2mm -8mm 0, clip]{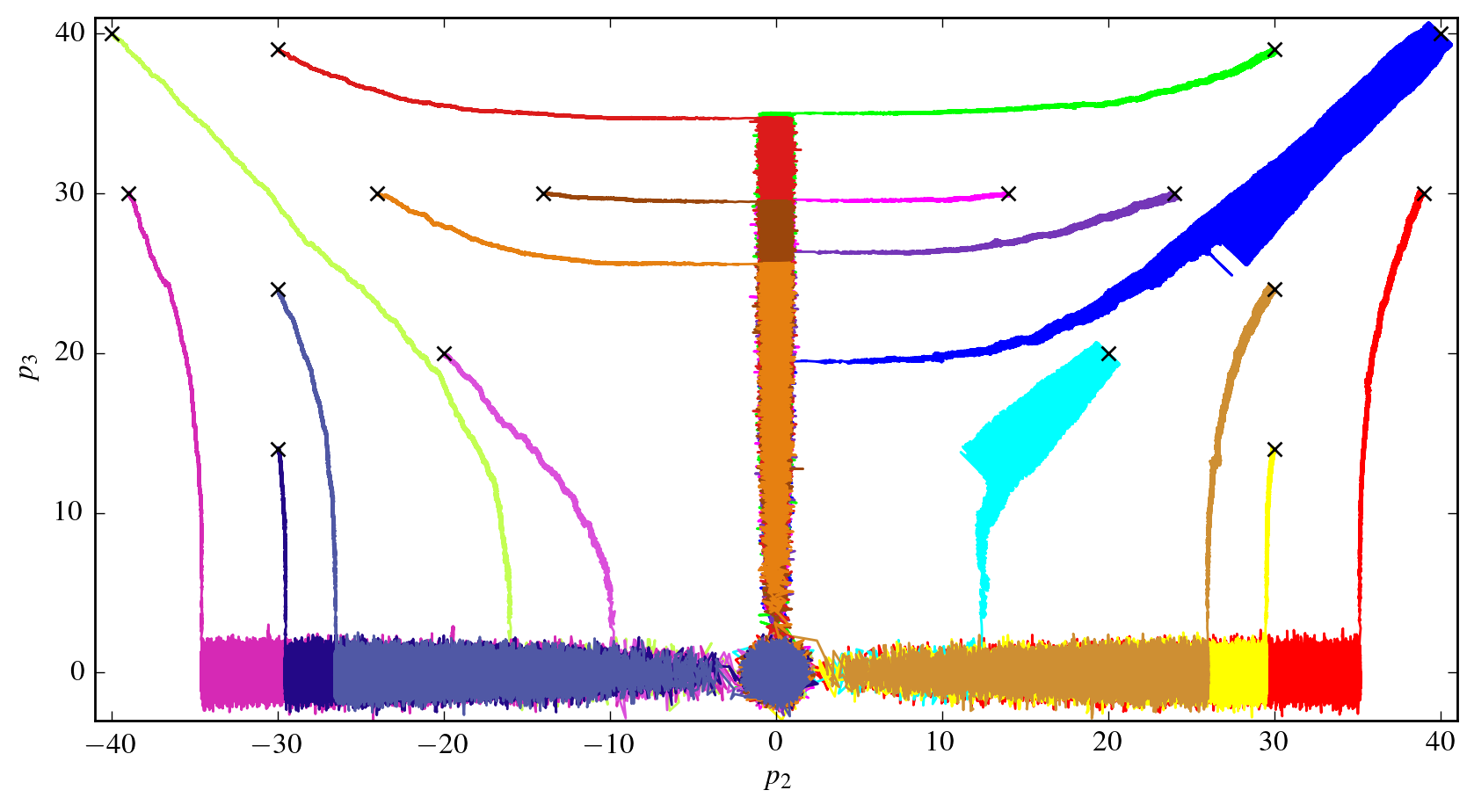}
\caption{The evolution of $p_2$ and $p_3$ for several initial
  conditions. The potentials are $W_\I = -\cos$,
$\I=\L, \C, \R$. Furthermore, $\gamma_1=\gamma_4 = 1$, $T_1 = 0.1$, and $T_4 =
0.4$.
Each ``$\times$'' sign indicates the beginning of a trajectory.
}\label{fig:traj_avg}
\end{figure}

The typical trajectory is as follows. Starting with some large $|p_2|$
and $|p_3|$, the slower of the two central rotors
is damped faster than the other, so that the projection drifts rapidly towards
one of the axes.
This leads to a regime where only one of the
central rotors is fast, while the other is essentially
thermalized. The energy in this fast rotor is gradually dissipated,
so that the orbit follows the axis towards the origin.

The behavior that we observe in \fref{fig:traj_avg} around the diagonal $p_2 =
p_3$
far enough from the origin is easily explained:
in the ``center of mass frame'' of the two central rotors,
we simply see two interacting rotors that oscillate slowly
in opposition, while being almost
decoupled from the outer rotors. More precisely, introducing
$Q = q_3-q_2\in \torus$ and $P=p_3-p_2\in \real$, we see that $(Q, P)$
acts approximately
as a mathematical pendulum with potential $2W_\C$, plus some rapidly
oscillating (and therefore weak) interactions with the outer rotors:
\begin{equ}
	\dot Q = P, \qquad \dot P = -2w_\C(Q) + \text{weak interactions}.
\end{equ}
Typically, if at first the energy in the center of mass frame
is not large enough to make a ``full turn,'' $Q$ oscillates
slowly around a minimum
of $W_\C$, which corresponds to a back-and-forth exchange of
momentum between 2 and 3, and explains the strips that
we observe around the diagonal.
The two central rotors are then gradually slowed down, 
until at some point the interaction with the external
rotors tears them apart.

The picture in the absence of noise (that is, when $T_1 = T_4 = 0$,
which is not covered by our assumptions)
is quite different, due to some resonances. We discuss
their nature in Appendix~\ref{s:resonances}.
These resonances are washed away by the noise, and are therefore not visible here.
They nevertheless play an important role in our computations, as we will see.

\mysubsection{Strategy}

In order to obtain rigorous results about the dynamics
and construct a Lyapunov function, we will apply specific
methods to each regime described above. We present them here in
increasing order of difficulty.

\begin{itemize}
\item When a significant part of the energy is contained in the outer rotors,
then as discussed above, the momenta of the two outer rotors essentially
decrease exponentially fast. In this region, the Lyapunov function
will be $e^{\theta H}$, and we will show that
when $p_1^2 + p_4^2$ is large enough and $\theta<
\min\left(1/{T_1}, 1/{T_4}\right)$, then
$Le^{\theta H} \lesssim -e^{\theta H}$ (\lref{l:LebetaH}).

\item When most of the energy is contained at just one of the central sites,
namely at site $j=2$ or $j=3$,
we will show that $L p_j \sim - p_j^{-3}$ when averaged
appropriately (\pref{prop:p1dansF1D1}). This corresponds to the neighborhood
of the axes in \fref{fig:traj_avg}.
This case is essentially treated as in \cite{cuneo_nonequilibrium_2014}.
In this region, we use a Lyapunov function
$V_j \sim e^{|p_j|^a + \frac\theta 2 p_j^2}$ (with
$a\in (0,1)$) such that $LV_j \lesssim -V_j/p_j^2$ (\pref{c:Lelambdap2tilde}).

\item When both $|p_2|$ and $|p_3|$ are large and hold most
of the energy, we do not
approximate the dynamics of $p_2$ and $p_3$ separately, but
we consider instead the ``central'' Hamiltonian
$H_c = \frac{p_2^2}{ 2} + \frac{p_3^2}{ 2} + W_{\L} + W_{\C}+ W_{\R}$.
We show that when averaged properly, $LH_c \sim - p_2^{-2} - p_3^{-2}$
(\pref{prop:technicalHb}). The Lyapunov function in this region
is $V_c \sim H_ce^{\theta H_c}$,
and we show (\pref{cor:borneOmegac}) that $LV_c \lesssim -V_c/H_c$.
Showing that $LH_c \sim - p_2^{-2} - p_3^{-2}$ is
the most difficult part of our proof.
The averaging of the rapidly oscillating forces will prove
to be insufficient due to some resonances, which
manifest themselves for some rational
values of $p_3/p_2$. We will consider separately the
vicinity of the $p_2 = p_3$ diagonal, which is easy to deal
with (\lref{lem:Hbloinresonances}), and the case
where $|p_3-p_2|$ is large, which requires substantially more
work (\sref{ss:decoupleddyn}). In the latter case,
we will use the rapid thermalization of the
external rotors in order to eliminate the resonant terms.

\end{itemize}

The factors $1/p_2^2$ and $1/H_c$
in $LV_j \lesssim -V_j/p_j^2$ and $LV_c \lesssim
-V_c/H_c$ are the cause of the 
logarithmic contribution in \eref{e:defphi}, which leads to the
subexponential convergence rate.

The final step (\sref{s:constrlyapunov}) is to combine $e^{\theta H}, V_2, V_3$
and $V_c$
(which each behave nicely in a given regime)
to obtain a Lyapunov function $V$ that behaves nicely everywhere
and  satisfies the conclusions of \tref{prop:Lyapunov}.

\mysubsection{The domains}\label{ss:thedomains}

Following the discussion above, we decompose $\Omega$ into several
sub-regions.
This decomposition only involves the momenta,
and not the positions.
All the sets in the decomposition are defined
in the complement of a ball $B_R$ of (large) radius $R$ in $p$-space:
\begin{equ}
B_R = \torus ^4 \times \Big\{ p \in \real^4: \sum_{i=1}^4 p_i^2 \le R^2\Big\}~.
\end{equ}

For convenience, we consider only $R\geq \sqrt{2}$ (see
\rref{rem:increasingOmega}).
We also use (large) integers $k$, $\ell$, and $m$
which will be fixed in \sref{s:constrlyapunov},
and we assume throughout that
\begin{equ}[eq:conditionklm]
1\leq  k < \ell < m~.	
\end{equ}

The first regions we consider are along the $p_2$ and $p_3$ axes:
\begin{equa}[e:omega23]
\Omega_2 &=\Omega_{2}(k, R)= \left\{x\in \Omega : p_2^2
>(p_1^2+p_3^2+p_4^2)^k\right\}\setminus B_{R}~,\\
\Omega_3 &=\Omega_{3}(k, R)= \left\{x\in \Omega : p_3^2 >
(p_1^2+p_2^2+p_4^2)^k\right\}\setminus B_{R}~.
\end{equa}
The region $\Omega_2$ (resp.~$\Omega_3$) corresponds to the configurations
where
most of the energy is concentrated at site 2 (resp.~3).
The next region corresponds to the configurations where
most of the energy is shared among the sites 2 and 3:
\begin{equ}\label{e:omegac}
\Omega_c = \Omega_{c}(\ell, m, R) = \left\{x\in \Omega :   p_2^2 + p_3^2 >
(p_1^2+p_4^2)^m ,~
p_3^{2\ell} >  p_2^2,~
p_2^{2\ell} > p_3^2\right\}\setminus B_{R}~
\end{equ}
(the conditions $p_3^{2\ell} >  p_2^2$ and $p_2^{2\ell} > p_3^2$
ensure that both $|p_2|$ and $|p_3|$ diverge sufficiently
fast when $\|p\|\to \infty$ in $\Omega_c$). 
These regions are illustrated in \fref{fig:intersectregionssphere}
and \fref{fig:intersectregionsp2p3}.
Note that $\Omega_2, \Omega_3, \Omega_c$ {\em do} intersect and {\em do not}
cover $\Omega$.
However, for $R$ large enough, the set $\Omega_2 \cup  \Omega_3 \cup \Omega_c \cup B_R$
contains the $p_2p_3$-plane (more precisely, the product
of $\torus^4$ and some neighborhood of the $p_2p_3$-plane in momentum space), which
is where the determining part of the dynamics lies, as discussed above.

\begin{figure}[ht]
\centering
\includegraphics[width=3.1in, trim = 0 1mm -0.5mm 0, clip]{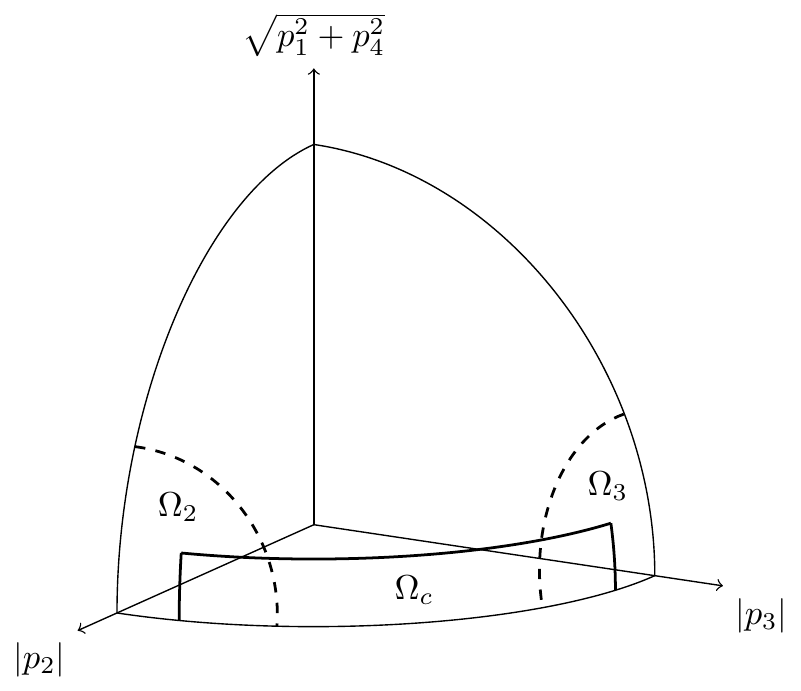}
\caption{A projection of the domains $\Omega_2, \Omega_3, \Omega_c$.
The spherical surface represents $\sum_{i=1}^4 p_i^2=C^2$ for some $C > R$.}
\label{fig:intersectregionssphere}
\end{figure}

\begin{figure}[ht]
\centering
\includegraphics[width=3.5in, trim = -3mm 0 0 0, clip]{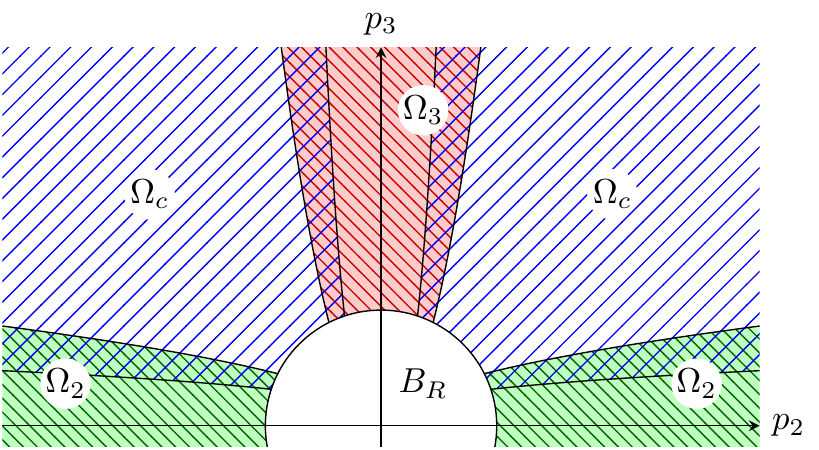}
\caption{The intersection of the sets $\Omega_2, \Omega_3, \Omega_c$
with the $p_2p_3$-plane (the lower half-plane is obtained
by axial symmetry).}
\label{fig:intersectregionsp2p3}
\end{figure}

\begin{remark}\label{rem:increasingOmega}
As a consequence of the restriction $R\geq \sqrt 2$, we have
$\Omega_j(k', R') \subset \Omega_j(k, R)$ for all $k'\geq k$,
$R'\geq R$, and $j=2,3$. Therefore, if a bound holds for all
$x\in \Omega_j(k, R)$, it also holds for all $x\in \Omega_j(k', R')$.
Similarly, at fixed $\ell$, we have
$\Omega_c(\ell, m', R')\subset \Omega_c(\ell, m, R)$
for all $m' \geq m$ and $R'\geq R$. This allows us to
increase $k, m$ and $R$ as needed (but not $\ell$).
We also observe immediately that 
for all $k, \ell, m$, and for $j=2,3$,
\begin{equa}
\lim_{R\to\infty}~\inf_{x\in \Omega_j(k, R)}|p_j| &=
\infty~,\label{eq:pjdivergeomegaj}\\
\lim_{R\to\infty}~\inf_{x\in \Omega_c(\ell, m, R)}|p_j| &=
\infty~.\label{eq:pjdivergeomegac}\\
\end{equa}
\end{remark}

\mysubsection{Notations}

Since averaging functions that rapidly oscillate in time
will play an important role, we introduce the $q_i$-average 
$\avg {f}_i = \frac {1}{2\pi}\int_0^{2\pi}f \dd q_i$
of a function $f:\Omega\to \real$ over one period of $q_i$. The result is
a function of $p$ and $\{q_j : j\neq i\}$.
In the presence of a generic function $f:\torus \to \real$ of one variable,
we write simply 
$\avg{f} =  \frac {1}{2\pi}\int_0^{2\pi}f(s) \dd s$, which is a constant.

For any function $f:\torus \to \real$ satisfying $\avg f = 0$,
one can find a unique integral $F: \torus \to \real$ such 
that $F' = f$ and $\avg F = 0$. More generally, we write
$f^{[j]}$ for the  $j^{\rm{th}}$ integral of $f$ that
averages to zero.

Without loss of generality,
we fix the additive constants of the potentials so that
\begin{equ}[eq:fixaddconst]
\avg{W_\I} = 0, \qquad \I=\L,\C,\R~.
\end{equ}
We also introduce two ``effective dissipation constants'':
\begin{equ}[eq:defalphas]
\alpha_2 = \gamma_1 \avg{W_{\L}^{2}}>0~, \qquad \alpha_3=\gamma_4\avg{
W_{\R}^{2}}>0~,
\end{equ}
where the positivity follows from \aref{as:assumptioncoupling}. Note also
that
because of
\eref{eq:fixaddconst}, there is no indeterminate additive constant in the
$\alpha_j$.

Finally, throughout the proofs, $c$ denotes a generic positive constant 
that can be each time different. These constants are allowed to depend on 
the parameters and functions at hand, but {\em not} on the position $x$.
We sometimes also use $c'$ to emphasize that the constant has changed.

\section{When only one of the central rotors is fast}\label{sec:23fast}

We consider the regime where either $|p_2|$ or $|p_3|$ (but not both) is much
larger than all other momenta.
The estimates for this regime are simple adaptations from
\cite{cuneo_nonequilibrium_2014}, but
we recall here the main ideas.

We start with some formal computations, thinking in terms
of powers of $p_2$ (resp. $p_3$) only. Then, we will restrict
ourselves to the set $\Omega_2(k, R)$  (resp. $\Omega_3(k, R)$) for
some large enough $k$ and $R$, so that the other momenta
are indeed ``negligible'' (see \lref{lem:kgrandOmega2})
compared to $p_2$ (resp. $p_3$).

\mysubsection{Averaging with one fast variable}\label{ssect:avgonevar}

Assume that $|p_2|$ is much larger than the other momenta.
We think in terms of the following fast-slow decomposition:
the variables $q_1, q_3, q_4$ and $p$ evolve slowly,
while $q_2$ evolves rapidly, since $\dot q_2 = p_2$,
and $p_2$ is large. In this regime, the variable
$q_2$ swipes through  $\torus$ many times
before any other variable changes significantly.
The dynamics for short times is
\begin{equs}[eq:approxdyn1fast]
p(t) & \approx p(0)~,\\
q_i(t) &\approx q_i(0)~, \qquad\qquad i=1,3,4~,	\\
q_2(t) &\approx  q_2(0) + p_2(0)t \pmod {2\pi}~.
\end{equs}

We consider an observable $f: \Omega\to \real$ and let $g$ be defined 
by
\begin{equ}[eq:beginaverage]
Lf = g~.
\end{equ}
Under the approximation \eref{eq:approxdyn1fast}, the quantity
$g(x(t))$ oscillates very rapidly
around its $q_2$-average $\avg{g}_2$, which is a function
of the slow variables $q_1, q_3, q_4$ and $p$.
We therefore expect the effective equation $Lf \approx  \avg{g}_2$ 
to describe the evolution of $f$ over several periods of
oscillations, and we now show how to give a precise meaning
to this approximation.

Although the stochastic terms (which appear as the second-order part
of the differential operator $L$) appear in the computations,
they do not play an important conceptual
role in this discussion; the rapid oscillations
that we average are of dynamical nature and
are present regardless of the stochastic forcing exerted by the
heat baths.

The generator of the dynamics \eref{eq:approxdyn1fast}
is simply $$L_2 = p_2 \partial_{q_2}~.$$
Decomposing the  generator $L$ defined in \eref{e:genl} as $L = L_2 + (L-L_2)$
and considering powers of $p_2$, we  view $L_2$
as large, and $L-L_2$ as small.
Note that for all smooth $h: \Omega \to \real$, we have 
$\avg{L_2 h}_2 =p_2\avg{\partial_{q_2} h}_2 = 0$ by periodicity,
so that the image of $L_2$
contains only functions with zero $q_2$-average.
Consider next the indefinite integral $G = \int (g -\avg{g}_2)\dd q_2$
(we choose the integration constant $C(q_1, q_3, q_4, p)$
to our convenience). By construction, we have
$L_2 (G/p_2) = g - \avg{g}_2$, so that
\begin{equ}[eq:avgscheme]
	L\left(f - \frac{G}{p_2}\right) = \avg{g}_2 +(L_2-L)\frac G{p_2}~.
\end{equ}
By subtracting the ``small'' counterterm  ${G}/{p_2}$ from $f$,
we have managed to replace $g$ with its $q_2$-average in the
right-hand side, plus some ``small'' correction. This procedure is
what we refer to as averaging with respect to $q_2$, and it
makes sense only in the regime where $|p_2|$ is very large.
If $\avg{g}_2 = 0$ and $(L_2-L)(G/p_2)$ is still oscillatory,
the procedure must be repeated.

\mysubsection{Application to the central momenta}

We now apply this averaging method to the observable $p_2$,
in the regime where $|p_2|$ is very large. By the definition of
$L$, we find
\begin{equ}[eq:startingLp2]
Lp_2 = w_\C -w_\L~.
\end{equ}
We have $\avg{w_\C}_2 = \avg{w_\L}_2 = 0$. Moreover,
$\partial_{q_2}W_\C(q_3-q_2) = -w_\C(q_3-q_2)$
and $\partial_{q_2}W_\L(q_2-q_1) = w_\L(q_2-q_1)$. Thus, in the notation
above, $G = \int (w_\C -w_\L)\dd q_2 = -W_\C - W_\L$, and we introduce the new variable
\begin{equ}\label{eq:p21}
p_2^{(1)} = p_2 - \frac G {p_2} = p_2 +\frac{W_\C + W_\L}{p_2}~.
\end{equ}
By \eref{eq:avgscheme}, we obtain
\begin{equs}[eq:Lp21]
Lp_2^{(1)} &= -(L_2-L)\left(\frac{W_\C + W_\L}{p_2} \right) \\
&= \frac{p_3w_\C-p_1w_\L}{p_2} + \frac{W_\C w_\L-W_\L w_\C +W_\L w_\L-W_\C
w_\C}{p_2^2}~.
\end{equs}
Observe that the right-hand side of \eref{eq:Lp21} is still
oscillatory, but now with an amplitude of order $1/p_2$, which is much
smaller than the amplitude of  \eref{eq:startingLp2} when $|p_2|$ is large.
Furthermore, the
right-hand side of \eref{eq:Lp21} has zero mean, since $\avg{w_\C}_2 = 
\avg{w_\L}_2 = 0$ and $$\avg{W_\C w_\L-W_\L w_\C +W_\L w_\L-W_\C w_\C}_2 = 
\frac 12 \avg{\partial_{q_2}(W_\C+W_\L)^2}_2 = 0~,$$
by periodicity.
In order to see a net effect, we need to average again. We consider
now the observable $f = p_2^{(1)}$, and apply the same procedure.
Instead of averaging the right-hand side of
\eref{eq:Lp21} in one step, we first deal only with the terms of order
$-1$ in $p_2$, by introducing
\begin{equ}\label{eq:p22}
p_2^{(2)} = p_2^{(1)}+{\frac
{p_1 W_{\L}+p_3 W_\C}{p_2^{2}}}~.
\end{equ}
We postpone further computations to the proof of \pref{prop:p1dansF1D1}
below, and explain here the main steps. We will see that
$Lp_2^{(2)}$ consists of terms of order $-2$ and $-3$ 
(by construction, the contribution of order $-1$ disappears).
The terms of order $-2$ have mean zero, and will be removed
by introducing a new variable $p_2^{(3)}$. We will then find
that $Lp_2^{(3)}$ contains terms of order $-3$ and $-4$.
To replace the terms of order $-3$ with their average
(which is finally non-zero), we will introduce a function
$p_2^{(4)}$. This will complete the averaging procedure.

We illustrate in \fref{fig:ave} the time-dependence
of $p_2, p_2^{(1)}$ and $p_2^{(2)}$ (slightly
shifted for better readability)\footnote{The irregularity
of the envelope of $p_2$ in \fref{fig:ave} is due to
the randomness of the phases of the two oscillatory forces
$w_\L$ and $w_\C$: they sometimes add up, and sometimes
compensate each other. Note also that the trajectory
of $p_2^{(2)}$ is rougher than the other two,
since the definition of  $p_2^{(2)}$ involves $p_1$,
which is directly affected by the stochastic force.}. Clearly, the
oscillations of $p_2^{(1)}$ are much smaller than
those of $p_2$, and we barely perceive the oscillations
of $p_2^{(2)}$, since they are smaller than the
random fluctuations. 

\begin{figure}[ht]
\centering
\includegraphics[width=0.8\linewidth, trim = 2mm 2mm -4mm 0, clip]{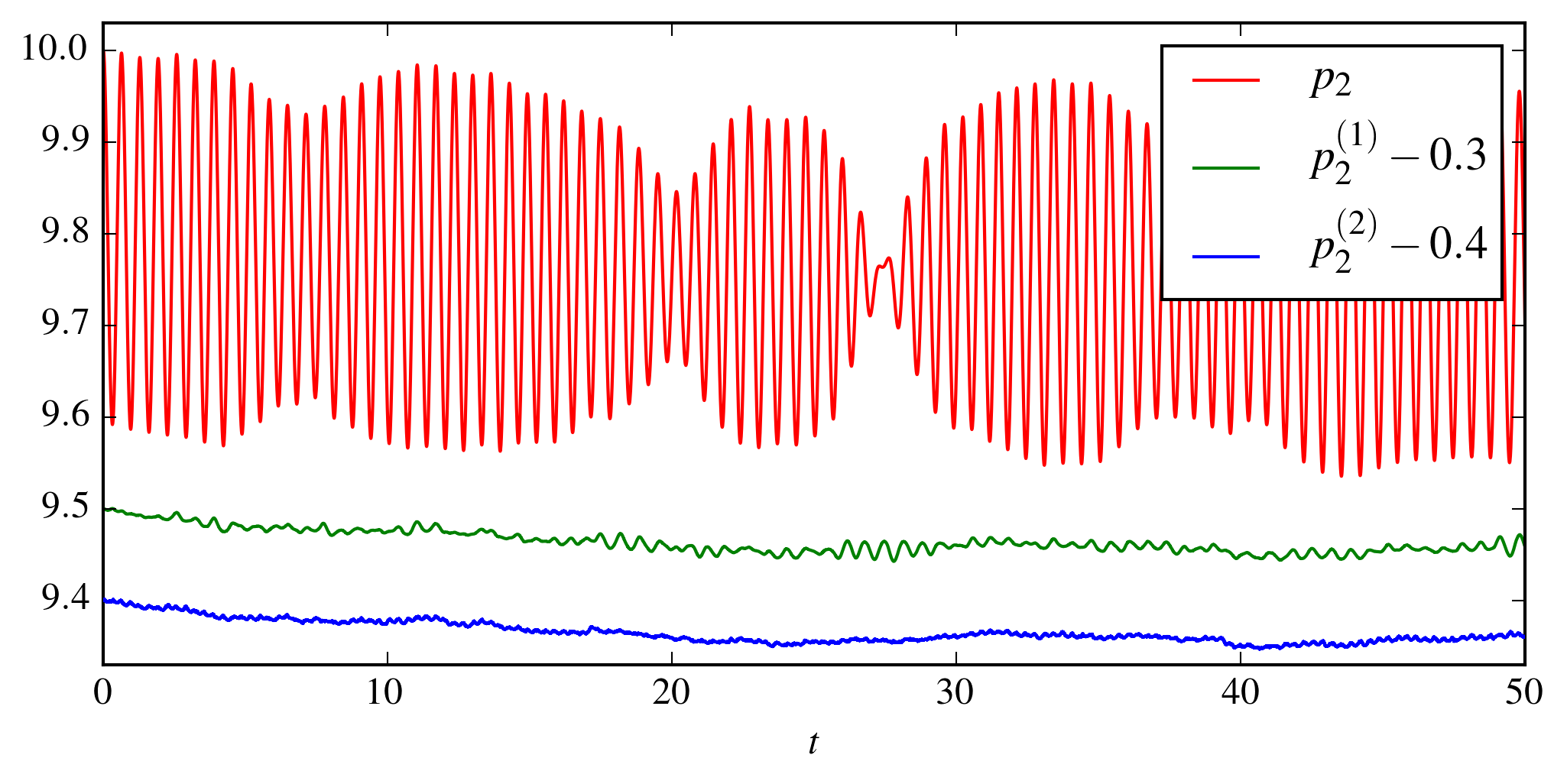}
\caption{The effect of the coordinate changes \eref{eq:p21} and
  \eref{eq:p22} on the effective
oscillations. Note that two of the curves are shifted vertically for easier
readability.
}\label{fig:ave}
\end{figure}

Before we state the result of this averaging process, 
we introduce a convenient notation for the remainders.

\begin{definition}\label{def:definition OOj}
Let $f, g$ be two functions defined on the set $\{x\in \Omega: p_2 \neq 0\}$.
We say that $f$ is $\OO_2(g)$ if there is a polynomial $z$
such that when $|p_2|$ is large enough,
\begin{equ}[eq:defOOj]
  |f(x)| \le z(p_1,p_3,p_4) | g(x)|~.
\end{equ}
The analogous notation $\OO_3$ will be used when $|p_3|$ is large,
and with a polynomial $z(p_1, p_2, p_4)$.
\end{definition}

This notation reflects the fact that when
most of the energy is at site 2 (resp. 3), one can forget
about the dependence on $p_1, p_3, p_4$ (resp. $p_1, p_2, p_4$),
provided that it is at most polynomial
(by the compactness of $\torus^4$, the position $q$ is irrelevant).
For example, the term ${
(p_1 W_{\L}+p_3 W_\C)/{p_2^{2}}}$ in \eref{eq:p22} is
$\OO_2(p_2^{-2})$.

It is easy to realize that the $\OO_j$, $j=2,3$, follow the same basic rules as
the usual $\OO $. In particular, $\OO_j(g_1) + \OO_j(g_2) = 
\OO_j(|g_1| + |g_2|)$ and  $\OO_j(g_1)\OO_j(g_2) = 
\OO_j(g_1g_2)$.

\begin{proposition}\label{prop:p1dansF1D1} There are functions $\pti2$ and
  $\pti3$ of the form
\begin{equa}
\pti2  &=p_2 + {\frac {W_{\L}+W_\C}{p_2}} +{\frac
{p_1 W_{\L} +p_3 W_\C}{p_2^{2}}}+ \OO_2(p_2^{-3})~,\label{eq:defp2tilde}\\
\pti3 &= p_3+ {\frac {W_{\R}+W_\C}{p_3}} +{\frac
{p_4 W_{\R}+p_2 W_\C}{p_3^{2}}} + \OO_3(p_3^{-3})~,\label{eq:defp3tilde}
\end{equa}
such that for $j=2,3$,
\begin{equ}\label{e:lp2}
L \pti j =-{\alpha_j}p_j^{-3} + \OO_j(p_j^{-4})~,
\end{equ}
where $\alpha _2>0$, $\alpha _3>0$ are defined in \eref{eq:defalphas}.
Furthermore,
\begin{equs}[3][eq:partderp2tilde]
\partial_{p_1}\pti2 &= \frac {W_{\L}}{p_2^2}+\OO_2(p_2^{-3})~, &\qquad
\partial_{p_4}\pti2 &= \OO_2(p_2^{-4})~,\\
\partial_{p_1}\pti3 &=\OO_3(p_3^{-4})~, &\qquad \partial_{p_4}\pti3 &=
\frac {W_{\R}}{p_3^2}+\OO_3(p_3^{-3})~.
\end{equs}

\end{proposition}
\begin{proof}It suffices to consider the case $j=2$.
The variable $\tilde p_2$ is constructed as in \cite{cuneo_nonequilibrium_2014}.
We continue the averaging procedure started above.
It is easy to check that $Lp_2^{(2)}$ can be written as
\begin{equ}
Lp_2^{(2)} = \frac {1}{{p_2^{2}}}\partial_{q_2} R_1 + 
 \frac{2 p_1\left(W_\L w_\L- W_\L w_\C\right) +
 2p_3\left( W_\C w_\L - W_\C w_\C \right)  }{p_2^3}~,
\end{equ}
with
\begin{equ}
R_1 =-p_1^{2} W_{\L}-p_3^{2} W_\C  -\gamma_1 p_1 \Wl1 +W_{\L}^{2}+W_{\L}
W_\C+W_{\C}^{2}+\Wc1  w_\R~.
\end{equ}
Since it is a total derivative, the term of order $-2$ has zero $q_2$-average, and
by introducing $
p_2^{(3)} = p_2^{(2)} - \frac {R_1} {{p_2^{3}}}
$, we find
\begin{equ}[eq:Lp23]
Lp_2^{(3)} = \frac { \partial_{q_2}R_2 +w_\L \Wl1 \gamma_1+ 
p_1 \left( W_\C w_\L-2 W_\L w_\C \right)+
 p_3\left( 2 W_\C w_\L-W_\L w_\C \right) }{p_2^{3}} + \OO_2(p_2^{-4})~,
\end{equ}
with
\begin{equs}
R_2 & = -p_1^{3}W_{\L} - p_3^{3} W_\C-3 
\gamma_1 p_1^{2}\Wl1 -\gamma_1^{2}p_1 \Wl2+3p_1 W_{\L}^{2}
\\&
\qquad +2 \gamma_1 T_1 \Wl1+(p_3-p_4)w'_\R \Wc2 + 3 p_3{W_\C^2}+3 p_3\Wc1  w_\R~.
\end{equs}
One can then average the terms of order -3 in \eref{eq:Lp23}. We have again
$\avg{  \partial_{q_2}R_2}_2 = 0$
by periodicity, and after integration by parts we find
\begin{equ}
	\avg{W_\C w_\L }_2 = \avg{w_\C W_\L}_2\quad \text{and} \quad \avg{\gamma_1 w_\L \Wl1}_2 =
-\gamma_1\avg{W_\L^2}_2 = -\alpha_2
\end{equ}
(for the signs, recall that $W_\L = W_\L(q_2-q_1)$
and $W_\C = W_\C(q_3-q_2)$). By adding appropriate counterterms
(not written explicitly), 
we obtain a function $p_2^{(4)} = p_2^{(3)} + \OO_2(p_2^{-4})$ such that
$$
Lp_2^{(4)} = - \frac{{\alpha_2}}{p_2^3} + \frac{\avg{p_3 W_\L w_\C- p_1W_\C 
w_\L}_2}{p_2^3}+ \OO_2(p_2^{-4})~.
$$
The first term in the right-hand side is the one we are looking for, 
and we deal with the other term of order ${-3}$ (which is non-zero)
as follows. We observe that
$$\avg{p_3 W_\L w_\C- p_1W_\C  w_\L}_2
= (p_1 \partial_{q_1} +p_3 \partial_{q_3})\avg{ W_\L W_\C}_2 = L\avg{ W_\L
W_\C}_2~,$$
since $\avg{ W_\L W_\C}_2$ is a function of $q_1, q_3$ only. We then set
$$\tilde p_2 =  p_2^{(4)} - \frac{\avg{ W_\L W_\C}_2}{p_2^3}$$
and obtain \eref{e:lp2}. It is immediate by the construction
of $\pti2$ that \eref{eq:partderp2tilde} holds.
\myqed\end{proof}

We now introduce a lemma, which says that
remainders of the kind $\OO_j(|p_j|^{-r})$, $j=2,3$,
can be made very small on $\Omega_j(k, R)$, 
provided that the parameters $k,R$ are large enough.

\begin{lemma}\label{lem:kgrandOmega2}Let $j\in \{2,3\}$ and $\pworkpower>0$.
Fix $\varepsilon > 0$ and a function $f = \OO_j(|p_j|^{-\pworkpower})$.
Then, for all sufficiently large $k$ and $R$, we have
$$
\sup_{x\in \Omega_j(k, R)} |f(x)| \leq \varepsilon~.
$$
\end{lemma}

\begin{proof} We prove the result for $j=2$.
By \dref{def:definition OOj} and \eref{eq:pjdivergeomegaj}, there is a
polynomial $z$ such that for all large enough $R$ and all $k$, we have
$|f|\leq  {z(p_1, p_3,p_4)}{|p_2|^{-\pworkpower}}$ on $\Omega_2(k, R)$.
But then, we have on the same set
$$
|f|\leq \frac {z(p_1, p_3,p_4)}{|p_2|^{\pworkpower}} \leq
\frac{c+c(p_1^2+p_3^2+p_4^2)^N}{|p_2|^\pworkpower} \leq \frac{c+c|p_2|^{\frac
{2N}k}}{|p_2|^\pworkpower} \leq c  |p_2|^{\frac {2N}k-\pworkpower}~,
$$
where
the second inequality is immediate for sufficiently large $N$,
the third inequality comes from the definition of $\Omega_2$, and the fourth
inequality holds because $|p_2|$ is bounded away from zero on $\Omega_2(k, R)$.
Recalling \eref{eq:pjdivergeomegaj}, we obtain the desired result when
$k$ is large enough so that $\frac {2N}k -r < -\frac r 2$.
\myqed\end{proof}

We now construct partial Lyapunov functions in the regions $\Omega_2$ and
$\Omega_3$.

\begin{proposition}\label{c:Lelambdap2tilde} Let $0<\theta < \min(1/T_1, 1/T_4)$
and $a\in (0,1)$. Consider the functions\footnote{
The role of the contribution $|\pti j|^a$ 
is to facilitate the patchwork that will lead to a global Lyapunov function in
\sref{s:constrlyapunov}. The corrections involving $F_2$ and $F_3$ help average
some $W_{\L}^2$ and $W_{\R}^2$
that appear in the computations. Without this correction, we would need a
condition on $\theta$ that is more restrictive than the natural condition
$\theta < \min(1/T_1, 1/T_4)$. }
\begin{equs}[eq:defVja]
V_2&= e^{|\pti 2|^a + \frac \theta 2  \ptipow
22}\left(1+F_2(q_2-q_1)/p_2^3\right)~,\\
V_3&= e^{|\pti 3|^a + \frac \theta 2  \ptipow
32}\left(1+F_3(q_3-q_4)/p_3^3\right)~,
\end{equs}
with the $\pti{j}$ of \pref{prop:p1dansF1D1}, and 
$F_2,F_3:\torus \to \real$ such that respectively $F_2'(s) =
\theta^2\gamma_1 T_1(\avg{W_{\L}^2}-W_{\L}^2(s) )$
and $F_3'(s) =
\theta^2\gamma_4 T_4(\avg{W_{\R}^2}-W_{\R }^2(s) )$.
Then, there are constants $\CstOmegaia,\CstOmegaib, \CstOmegaid>0$, independent
of $a\in (0,1)$, such that for all sufficiently large $k$ and $R$,
we have for $j=2,3$ the following inequalities on $\Omega_j$:
\begin{equ}\label{e:ptij}
\CstOmegaia e^{|p_j|^a + \frac \theta 2  p_j^2} < V_j < \CstOmegaib
e^{|p_j|^a + \frac \theta 2  p_j^2}~,
\end{equ}
\begin{equ}[e:Lfp2tilde]
L  V_j \leq -\CstOmegaid
p_j^{-2}e^{|p_j|^a + \frac \theta 2  p_j^2}~.
\end{equ}
\end{proposition}

\begin{proof}By symmetry, it suffices to prove
the result for $j=2$. In this proof, we do not allow the $\OO_2$ to depend on
$a\in (0,1)$ (that is, we want the bound \eref{eq:defOOj} to hold uniformly in
$a\in (0,1)$). We start by proving \eref{e:ptij}. For large enough $R$, we
have that $|p_2| > 2$ on $\Omega_2$. Moreover,
since $\pti 2 = p_2 + \mathcal O_2(p_2^{-1})$ and $F_2(q_2-q_1)/p_2^3 =
\mathcal O_2(p_2^{-3})$,
we have by \lref{lem:kgrandOmega2}
that for large enough $k, R$, it holds on $\Omega_2$ that
\begin{equ}[eq:pti2geq1]
|\pti 2| > 1 \quad \text{and}\quad  \left|\frac{F_2(q_2-q_1)}{p_2^3}\right|<
\frac 12 ~.
\end{equ}
Moreover, since both $|\pti 2|$ and
$|p_2|$ are $>1$, \eref{eq:defp2tilde} implies, for all $a\in (0,1)$,
\begin{equs}
||\pti 2|^a - |p_2|^a|\leq |\ptipow 22 - p_2^2| = \left|2({W_{\L}+W_\C}) +
\mathcal O_2(p_2^{-1})\right|~.
\end{equs}
Since  $W_{\L}$ and $W_\C$ are bounded, it follows from \lref{lem:kgrandOmega2}
that we can bound the right-hand side by a constant, so that
we find
\begin{equ}\label{e:exptildeetc}
ce^{|p_2|^a + \frac \theta 2  p_2^2} <  e^{|\pti 2|^a + \frac \theta 2  \ptipow
22} < c' e^{|p_2|^a + \frac \theta 2  p_2^2}~,
\end{equ}
uniformly in $a$.
By this, by the definition of $V_2$, and by \eref{eq:pti2geq1}, we obtain
\eref{e:ptij}. We now prove \eref{e:Lfp2tilde}. Let $f(s) =
e^{|s|^a + \frac \theta 2  s^2}$ and note that
\begin{equs}[eq:LV2calc1]
L \big(e^{|\pti2|^a + \frac \theta 2  \ptipow22}\big) & = L f(\pti2)
= f'(\pti2) L\pti2 + f''(\pti2) \sum_{b=1,4}\gamma_b T_b
(\partial_{p_b}\pti2)^2~.
\end{equs}
By \pref{prop:p1dansF1D1}, we have on $\Omega_2$ that
\begin{equs}[eq:LV2calc2]
f'(\pti2) L\pti2 & = e^{|\pti2|^a + \frac \theta 2  \ptipow22}\left(\frac
{a|\pti2|^a}{\pti2}+\theta \pti2\right)\left(-{\alpha_2}p_2^{-3} +  \mathcal
O_2(p_2^{-4})\right)\\
&= e^{|\pti2|^a + \frac \theta 2  \ptipow22}\left(-\alpha_2\theta p_2^{-2} +
\OO_2(p_2^{-3})\right)~,
\end{equs}
where we have used that $\pti2 = p_2 + \OO_2(p_2^{-1})$, and that
$a|\pti2|^{a-1} < 1$ (since $|\pti2| > 1$), so that on $\Omega_2$, the
$\OO_2(p_2^{-3})$ obtained is indeed uniform in $a$.
Next, one can verify that uniformly in $a\in (0,1)$ and $|\pti2|> 1$,
\begin{equs}
f''(\pti2) 
&\leq f(\pti2)\left(\theta^2 \ptipow 22 + 2\theta |\pti2| + c\right)~.
\end{equs}
Moreover,
by \eref{eq:partderp2tilde} we have $ \sum_{b=1,4}\gamma_b T_b
(\partial_{p_b}\pti2)^2  = \gamma_1T_1  {W_{\L}^2}/{p_2^4} + \OO_2
(p_2^{-5})$, so that on $\Omega_2$,
\begin{equs}[eq:LV2calc3]
f''(\pti2) \sum_{b=1,4}\gamma_b T_b (\partial_{p_b}\pti2)^2 &\leq e^{|\pti2|^a
+
\frac \theta 2  \ptipow22}\left(\theta^2 \ptipow22 + 2\theta |\pti2| +
c\right)\left(\gamma_1T_1 \frac {W_{\L}^2}{p_2^4} + \OO_2
(p_2^{-5})\right)\\
& =e^{|\pti2|^a + \frac \theta 2  \ptipow22}\left( \theta^2\gamma_1T_1 \frac
{W_{\L}^2}{p_2^2} + \OO_2 (p_2^{-3})\right)~.
\end{equs}
Therefore, by \eref{eq:LV2calc1}, \eref{eq:LV2calc2} and \eref{eq:LV2calc3},
\begin{equs}
L \big(e^{|\pti2|^a + \frac \theta 2  \ptipow22}\big) & \leq
\frac 1 {p_2^2}e^{|\pti2|^a + \frac \theta 2  \ptipow22}\left(-\alpha_2\theta +
\theta^2\gamma_1T_1 W_{\L}^2 + \OO_2(p_2^{-1}) \right)~.
\end{equs}
But then
\begin{equs}
L V_2 &=  \left(1+\frac {F_2(q_2-q_1)}{p_2^3}\right)L\left(e^{|\pti 2|^a + \frac \theta 2  \ptipow
22}\right) + e^{|\pti 2|^a + \frac
\theta 2  \ptipow 22}L\left(\frac {F_2(q_2-q_1)}{p_2^3}\right)\\
& \leq \left(1+\frac {F_2(q_2-q_1)}{p_2^3}\right)\frac 1 {p_2^2}e^{|\pti2|^a + \frac \theta 2  \ptipow22}\left(-\alpha_2\theta + \theta^2\gamma_1T_1 W_{\L}^2 + \OO_2(p_2^{-1}) \right) \\
&\qquad\qquad + e^{|\pti 2|^a + \frac
\theta 2  \ptipow 22}\left(\frac {\theta^2\gamma_1T_1 \big(\avg{W_{\L}^2}_2 - W_{\L}^2\big)}{p_2^2} + \OO_2(p_2^{-3})\right)\\
& = 
\frac 1 {p_2^2}e^{|\pti2|^a + \frac \theta 2  \ptipow22}\left(-\alpha_2\theta +
\theta^2\gamma_1T_1 \avg{W_{\L}^2} + \OO_2(p_2^{-1}) \right)~.
\end{equs}
Using the definition of $\alpha_2$ in \eref{eq:defalphas} and the
condition on $\theta$, we find that $-\alpha_2\theta +
\theta^2\gamma_1T_1 \avg{W_{\L}^2}$ is negative. Using then
\lref{lem:kgrandOmega2}
to make the $ \OO_2(p_2^{-1})$ very small, and combining the result with
\eref{e:exptildeetc} completes the proof.
\myqed\end{proof}

\section{When both central rotors are fast}\label{sec:bothfast}

We now study the regime where {\em both} $|p_2|$ and $|p_3|$ are large (not necessarily
of the same order of magnitude), and
$|p_1|$ and $|p_4|$ are much smaller. We then have two fast variables: $q_2$
and $q_3$. As we will see, this will lead to some trouble related to
resonances, and averaging the rapid oscillations will not be enough.
We start with some formal computations thinking in terms of 
powers of $p_2$ and $p_3$, and then restrict ourselves to
the set $\Omega_c(\ell, m, R)$ for some appropriate parameters.

\mysubsection{Averaging with two fast variables: resonances}\label{ss:averagingtwofast}

Now the fast-slow decomposition is as follows: $q_1, q_4$
and $p$ are the slow variables, and $q_2,q_3$ are the
fast variables, with the approximate dynamics (for short times)
\begin{equs}[eq:fastslowtwofast]
p(t) & \approx p(0)~,\\
q_i(t) &\approx q_i(0)~, \qquad\qquad i=1,4~,	\\
q_2(t) &\approx  q_2(0) + p_2(0)t \pmod {2\pi}~, \\
q_3(t) &\approx  q_3(0) + p_3(0)t \pmod {2\pi}~,
\end{equs}
generated by  $L_2 + L_3 = p_2\partial_{q_2}+p_3\partial_{q_3}$,
which we see as the most important contribution in $L$.
Let again $f,g:\Omega \to \real$ and assume that
$$
Lf = g~.
$$
We would like, as above, to add a correction to $f$
in the left-hand side in order to replace $g$ with its
average in the right-hand side.
However, since the fast motion of $(q_2,q_3)$ on $\torus^2$
(in the dynamics \eref{eq:fastslowtwofast}) follows orbits that
are open or closed depending on whether $p_2$
and $p_3$ are commensurable or not,
there seems to be no natural notion of ``average of $g$'' that
is continuous with respect to the slow variables.
 
Consider for example $g(x) = \sin(2q_2-q_3)$. In our approximation, 
$\sin(2q_2(t)-q_3(t))$ oscillates with frequency
$(2p_2-p_3)/2\pi $. The average is zero when $p_3 \neq 2p_2$,
and $\sin(2q_2(t)-q_3(t))$ remains constant when $p_3 = 2p_2$.
When $p_3$ is close to $2p_2$, the oscillations are slow,
and one cannot simply average $\sin(2q_2(t)-q_3(t))$.
More generally, any smooth function $g$ on $\Omega$ can be written as
$\sum_{n, m\in \mathbb Z}a_{n,m}\sin(nq_2 + m q_3 + \phi_{n,m})$ 
for some coefficients $a_{n,m}$ and $\phi_{n,m}$
which depend on the slow variables $q_1, q_4$ and $p$.
Each such term
gives rise to problems close to the line $p_3/p_2 = - n/m$
in the $p_2p_3$-plane.

However, if
$g$ depends on $q_2$ but not on $q_3$, then no problem appears. In the
approximation \eref{eq:fastslowtwofast}, the quantity
$g(x(t))$ then oscillates rapidly around $\avg{g}_2$, which is
then a function of the slow variables $q_1, q_4$ and $p$. Then,
as in \sref{ssect:avgonevar},
we use $G = \int(g -\avg{g}_2)\dd q_2$ (we choose the integration
constant independent of $q_3$), so that
$(L_2 + L_3)(G/p_2) = L_2(G/p_2) =  g -\avg{g}_2$. Thus,
$L(f- G/p_2) = \avg{g}_2 + (L-L_2-L_3)(f- G/p_2)$, which has the
desired form. Similarly,
if $g$ depends on $q_3$ but not on $q_2$, we use the counterterm
$G/p_3$ with $G = \int (g -\avg{g}_3)\dd q_3$.
And of course, if $g$ can be decomposed as the sum of a function
not involving $q_3$ and a function not involving $q_2$, then
we can average each part separately and sum the two counterterms.

It turns out that we will mostly encounter terms that depend only on one of the
fast variables, and are therefore easy to average. We will go as far as
possible averaging such terms, and then introduce a method
to deal with the resonant terms (involving both $q_2$ and $q_3$) that appear.

\mysubsection{Application to the central energy}

As a starting point, we use the central energy
$$
H_c = \frac{p_2^2}{ 2} + \frac{p_3^2}{ 2} + W_{\L} + W_{\C}+ W_{\R}~.
$$

\begin{definition}\label{def:definition OOc} Let $A_* \equiv
\{x\in \Omega: p_2 \neq 0, p_3\neq 0\}$ and let $f, g$ be two functions
defined on a  set $A\subset A_*$. We say that $f$ is
$\OO_c(g)$ (on the set $A$) if there is a polynomial $z$ such that for all
$x\in
A$ with $\min(|p_2|, |p_3|)$ large enough, we have
  \begin{equ}
    |f(x)| \leq z(p_1,p_4) |g(x)|~.
  \end{equ}
Unless explicitly stated otherwise, we take $A = A_*$. 
\end{definition}

We state the main result of this section.

\begin{proposition}\label{prop:technicalHb} There is a function of the form
\begin{equ}[eq:deftHc]
\tH_c = H_c + {\frac {p_1W_{\L}}{p_2}} +{\frac {p_4W_{\R}}{p_3}} +
\OO_c(|p_2|^{-2}+|p_3|^{-2})~,
\end{equ}
such that
\begin{equ}
\label{e:Lhbn23}
L\tH_c =-{\frac {\alpha_2 }{  p_2^{2}}}-{\frac {\alpha_3}{p_3^{2}}}+
\OO_c(|p_2|^{-5/2}+|p_3|^{-5/2})~,
\end{equ}
with $\alpha_j$ as defined in \eref{eq:defalphas}.
Furthermore,
\begin{equ}[e:Hctildep1p4derivatives]
\partial_{p_1}\tH_c = \frac {W_{\L}}{p_2}+\OO_c(p_2^{-2})~, \qquad
\partial_{p_4}\tH_c = \frac {W_{\R}}{p_3}+\OO_c(p_3^{-2})~.
\end{equ}
\end{proposition}

In order to reduce the length of some symmetric formulae,
we use the notation ``$+ \lr$'' as a shorthand for the other half
of the terms with the indices exchanged as follows: $1\lr 4$, $2\lr 3$,
$\L \lr \R$, and the sign of $w_\C$ changed (due to the 
asymmetry of the argument $q_3-q_2$ of $W_\C$).

In order to prepare the proof of \pref{prop:technicalHb}, we proceed
as follows. We first see that
\begin{equ}
L(H_c) = -p_1 w_{\L} - p_4 w_{\R}~.
\end{equ}
Since $w_\L$ does not involve $q_3$ and $w_\R$ does not involve $q_2$, it 
easy to find appropriate counterterms: we introduce
\begin{equa}[e:HC1]
\Hc1 &= H_c + \frac{p_1 W_{\L}}{ p_2} +  \frac{p_4 W_{\R}}{p_3}~,
\end{equa}
and obtain 
\begin{equ}
L\Hc1 = 
{\frac {-\gamma_1 p_1W_{\L}  
-  p_1^{2}w_{\L} 
+w_{\L}  W_{\L}  }{p_2}}+\frac {p_1 W_{\L}(w_\L -w_\C )}{p_2^{2}} 
~+ \lr.
\end{equ}
The terms of order $1/p_2$ do not depend on $q_3$
and have mean zero with respect to $q_2$
(again $w_\L W_\L =\partial_{q_2}W_\L^2/2$ has zero $q_2$-average by
periodicity). Similarly, the terms in $1/p_3$ do not involve $q_2$
and average to zero with respect to $q_3$.
Therefore, we introduce a next round of counterterms:
\begin{equ}
\Hc2 = \Hc1 + \left(\frac{\gamma_1p_1\Wl1
+p_1^2 W_{\L}-W_{\L}^{2}/2}{p_2^{2}}~ + \lr\right)~,
\end{equ}
and obtain 
\begin{equs}[eq:Lhc2] 
L\Hc2 &= \left(\frac{-p_1^3w_\L-3\gamma_1 p_1^2W_\L-\gamma_1^2p_1W_L^{[1]}+ 4p_1
W_\L w_\L+2\gamma_1 T_1 W_\L}{p_2^2} ~+ \lr  \right)\\
&  \qquad + \frac{ \gamma_1 w_\L W_\L^{[1]} }{p_2^2}+ \frac{ \gamma_4 w_\R 
W_\R^{[1]} }{p_3^2} -{\frac {p_1W_\L w_\C    }{p_2^{2}}}+{\frac {
p_4W_\R  w_\C}{p_3^{2}}}+ \mathcal \OO_c(|p_2|^{-3} +|p_3|^{-3})~.
\end{equs}
The terms in the first line are easy to eliminate, since each one depends on
only one of the fast variables and averages to zero. The terms $\gamma_1 w_\L
W_\L^{[1]}/p_2^2$ and $\gamma_4 w_\R  W_\R^{[1]}/p_3^2$ are the ones we are
looking for, since after integrating
by parts, we find $\big\langle \gamma_1 w_\L  W_\L^{[1]}\big\rangle_2 =
-\gamma_1\big\langle W_\L^2\big\rangle_2 =
-\alpha_2$ and $\big\langle \gamma_4 w_\R  W_\R^{[1]}\big\rangle_3 =
-\gamma_4\big\langle W_\R^2\big\rangle_3 = -\alpha_3$.
The two ``resonant'' terms
involving $W_\L w_\C$ and $W_\R  w_\C$ are more problematic and we leave them
untouched for now. By introducing the appropriate counterterms (which we
do not write explicitly), we obtain a function $
\Hc3 = \Hc2 + \OO_c(|p_2|^{-3}+|p_3|^{-3})
$ such that
\begin{equ}[e:round3]
L\Hc3 =-{\frac {\alpha_2}{p_2^{2}}}-{\frac {
\alpha_3}{p_3^{2}}} -{\frac {p_1W_\L w_\C    }{p_2^{2}}}+{\frac {
p_4W_\R  w_\C}{p_3^{2}}}+  \OO_c(|p_2|^{-3}+|p_3|^{-3})~.
\end{equ}

In order to obtain \eref{e:Lhbn23}, we must get rid of the 
two ``mixed'' terms involving $W_\L w_\C$ and $W_\R w_\C$, which
are of the same order as the dissipative contributions involving
$\alpha_2$ and $\alpha_3$.
Since they each depend on both $q_2$ and $q_3$, these terms are not easy to get
rid of, due to the resonance phenomenon discussed above. 
In fact, as discussed in Appendix~\ref{s:resonances},
these resonances have a physical meaning. 
Their effect becomes clearly visible when $T_1 = T_4=0$ (which is not covered by
our assumptions): they alter the dynamics
in the $p_2p_3$-plane, but do not prevent $H_c$ from decreasing in average.
We postpone to \sref{ss:decoupleddyn} the construction of the counterterms
that will eliminate these resonant terms. 

We introduce next two technical lemmata and an application of
\pref{prop:technicalHb}. The following lemma
is analogous to \lref{lem:kgrandOmega2}.
\begin{lemma}\label{lem:mgrandOmegac}Let $j\in \{2,3\}$ and $\pworkpower>0$.
Fix an integer $\ell > 0$, and an $\varepsilon > 0$. Let $f$ be some
$\OO_c(|p_j|^{-\pworkpower})$ on the
set $A_* = \{x\in \Omega: p_2 \neq 0, p_3\neq 0\}$. Then, for all
sufficiently large $m$ and $R$, we have
$$
 \sup_{x\in \Omega_c(\ell,m, R)} |f(x)| < \varepsilon~.
$$
\end{lemma}
\begin{proof}We prove the result for $f = \OO_c(|p_2|^{-\pworkpowerd})$ and
proceed
as in \lref{lem:kgrandOmega2}.
By \dref{def:definition OOc} and \eref{eq:pjdivergeomegac}, there is a
polynomial
$z$ such that for all $m$ and all sufficiently large $R$, we have on
$\Omega_c(\ell, m, R)$,
\begin{equs}
|f| &\leq \frac {z(p_1, p_4)}{|p_2|^{\pworkpowerd}} \leq
\frac{c+c(p_1^2+p_4^2)^N}{|p_2|^\pworkpowerd} \leq \frac{c+c(p_2^2
+p_3^2)^{\frac {N}m}}{|p_2|^\pworkpower}\\
& \leq \frac{c+(p_2^2 +p_2^{2\ell})^{\frac {N}m}}{|p_2|^\pworkpower}\leq
\frac{c+c|p_2|^{\frac {2\ell N}m}}{|p_2|^\pworkpower}\leq c|p_2|^{\frac {2\ell
N}m-r}~,
\end{equs}
where we choose $N$ large enough and use
the definition of $\Omega_c$. By \eref{eq:pjdivergeomegac},
we conclude that the desired result holds
for $m$ large enough so that $\frac {2\ell N}m-r < -\frac r2$.
\myqed\end{proof}

\begin{lemma}\label{lem:repartitionpowers} Let  $ f = \mathcal
O_c(p_2^{z_1}p_3^{z_2})$ for some $z_1, z_2 \in  \mathbb  R$ such that $z_1,
z_2$ have the same sign.
Then,
 $f =
\OO_c(|p_2|^{z_1+z_2} + |p_3|^{z_1+z_2})$.
\end{lemma}
\begin{proof}
We apply Young's inequality in the form $xy \leq x^a + y^b$ with
$a = \frac
{z_1+z_2}{z_1} > 1$, $b = \frac{z_1+z_2}{z_2}>1$,
and  $x = |p_2|^{z_1}$, $y = |p_3|^{z_2}$. We obtain
$|p_2|^{z_1}|p_3|^{z_2} \leq |p_2|^{z_1+z_2}+|p_3|^{z_1+z_2}$. 
This, and the definition of $\OO_c$, complete the proof. 
\myqed\end{proof}

As a consequence of \pref{prop:technicalHb} we have:

\begin{proposition}\label{cor:borneOmegac}
Let $0<\theta < \min(1/T_1, 1/T_4)$ and define
$$
V_c=\tH_ce^{ \theta  \tH_c}\left(1+ \frac{F_2(q_2-q_1)}{p_2^3}+
\frac{F_3(q_3-q_4)}{p_3^3}
\right)~,
$$
with the $\tH_c$ of \pref{prop:technicalHb} and $F_2$, $F_3$ as in
\pref{c:Lelambdap2tilde}. Let $\ell>1$ be a fixed integer.
Then, there are constants
$\CstOmegaca,\CstOmegacb, \CstOmegacd>0$ such that
for all large enough $m$ and $R$, the following inequalities
hold on $\Omega _c(m, \ell,R)$:
\begin{equ}\label{e:bornesVcOmegac}
\CstOmegaca (p_2^2 + p_3^2)e^{\frac \theta 2 (p_2^2 + p_3^2)}< V_c< \CstOmegacb
(p_2^2 + p_3^2)e^{\frac \theta 2 (p_2^2 + p_3^2)}~,
\end{equ}
\begin{equ}[e:bornesLVcOmegac]
L  V_c \leq -\CstOmegacd e^{\frac \theta 2 (p_2^2 + p_3^2)}~.
\end{equ}
\end{proposition}
\begin{proof}
We first prove \eref{e:bornesVcOmegac}. By \eref{eq:deftHc}, the boundedness of
the potentials, and \lref{lem:mgrandOmegac},
we have for $m, R$ large enough that on $\Omega_c$,
\begin{equ}[e:borneHctildeHc]
\left|\tH_c - \frac {p_2^2}2 - \frac {p_3^2}2\right|< c \qquad \text{and}\qquad
\left|\frac{F_2(q_2-q_1)}{p_2^3}+ \frac{F_3(q_3-q_4)}{p_3^3}\right|< \frac 12~.
\end{equ}
In addition, if $m, R$ are large enough, $p_2^2 + p_3^2$
is large on $\Omega_c$, so that the first part of \eref{e:borneHctildeHc} implies that
$c(p_2^2 + p_3^2)<\tH_c <c'(p_2^2 + p_3^2)$. This and \eref{e:borneHctildeHc}
imply \eref{e:bornesVcOmegac}.

We next prove \eref{e:bornesLVcOmegac}. Define $f(s) = se^{\theta s}$.
By \pref{prop:technicalHb},
\begin{equs}[eq:LHcethetaHcfirst]
L \big(\tH_c e^{ \theta  \tH_c}\big)  & = L f(\tH_c)
= f'(\tH_c) L \tH_c + f''(\tH_c) \sum_{b=1,4}\gamma_b T_b
(\partial_{p_b}\tH_c)^2 \\
& = e^{ \theta  \tH_c} \left(\theta \tH_c + 1\right)\left(-{\frac
{\alpha_2 }{ p_2^{2}}}-{\frac {\alpha_3}{p_3^{2}}}+  \OO_c
(|p_2|^{-5/2}+|p_3|^{-5/2})\right) \\
&\quad +  e^{ \theta  \tH_c} \left(\theta^2 \tilde
H_c+2\theta\right)\left(\gamma_1 T_1 \frac {W_{\L}^2}{p_2^2} +\gamma_4 T_4
\frac {W_{\R}^2}{p_3^2} +   \OO_c (|p_2|^{-3}+|p_3|^{-3})\right)~.
\end{equs}
Now observe that for any $C\in \real$, we have
$$
\tH_c + C =   \frac{p_2^2+p_3^2}2 + \OO_c(1) =\frac{p_2^2+p_3^2}2\left( 1 + \OO_c(p_2^{-2}+p_3^{-2})\right),
$$
since trivially $(p_2^2+p_3^2)^{-1} \leq p_2^{-2} + p_3^{-2}$.
But then, by \eref{eq:LHcethetaHcfirst} and  \lref{lem:repartitionpowers},
we find that
\begin{equs}
L \big(\tH_c e^{ \theta  \tH_c}\big)  & \leq e^{ \theta  \tH_c}
 \frac{p_2^2+p_3^2}2\left(\frac{\theta^2\gamma_1T_1
W_{\L}^2-\theta\alpha_2 + \OO_c(|p_2|^{-1/2})}{p_2^2} ~+\lr\right)~.
\end{equs}
As in the proof of \pref{c:Lelambdap2tilde}, the corrections involving
$F_2$ and $F_3$ replace the oscillatory terms $W_\L^2$ and $W_\R^2$
with their averages:
\begin{equs}
L V_c &=  L\left(\tH_ce^{ \theta  \tH_c}\right)\left(1+\frac {F_2}{p_2^3}+\frac
{F_3}{p_3^3}\right) + \tH_ce^{ \theta  \tH_c}L\left(\frac {F_2}{p_2^3}+\frac
{F_3}{p_3^3}\right)\\
& \leq e^{ \theta  \tH_c}
 \frac{p_2^2+p_3^2}2\left(\frac{\theta^2\gamma_1T_1
\avg{W_{\L}^2}-\theta\alpha_2 + \OO_c(|p_2|^{-1/2})}{p_2^2} ~+\lr\right)~.
\end{equs}
Therefore, by the definition \eref{eq:defalphas} of $\alpha_j$
and the condition on $\theta$, we have
\begin{equs}
L V_c & \leq e^{ \theta  \tH_c}
 \frac{p_2^2+p_3^2}2\left(\frac{-c + \OO_c(|p_2|^{-1/2})}{p_2^2}+
 \frac{-c + \OO_c(|p_3|^{-1/2})}{p_3^2}\right)~.
\end{equs}
Finally, by \lref{lem:mgrandOmegac},
and using that $(p_2^2+p_3^2)(p_2^{-2}+p_3^{-2}) > 2$, we
indeed obtain \eref{e:bornesLVcOmegac}.
\myqed\end{proof}

We now return to the proof of \pref{prop:technicalHb}. We need to
find some counterterms to eliminate the mixed terms in \eref{e:round3}.
For this, we use a subdivision of $A_* = \{x\in \Omega: p_2\neq 0, p_3\neq
0\}$ into 3 disjoint pieces, as shown in \fref{fig:intersectregionssphere2}:
\begin{equa}[e:subregions]
\OA&=\{x\in A_*:|p_2+p_3| \geq (p_2-p_3)^2 \}~,\\
\OB&=\left\{x\in A_* :(p_2-p_3)^2 > |p_2+p_3| > (p_2-p_3)^2/2
\right\}~,\\
\OC& =\left\{x\in A_*: (p_2-p_3)^2 \geq 2|p_2+p_3| \right\}~.
\end{equa}

By construction, $A_1$ is close to the diagonal $p_2=p_3$,
$A_3$ is far from it, and $A_2$ is some transition region.

\begin{figure}[ht]
\centering
\includegraphics[width=2.5in]{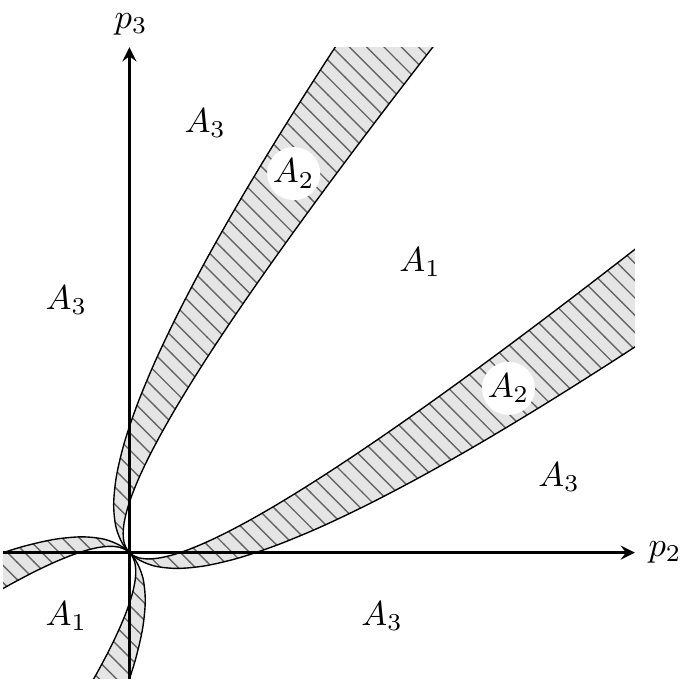}
\caption{Projection of the partition $A_*= \OA \cup \OB \cup \OC$ onto the
$p_2p_3$-plane. Note that the sets $A_1, A_2$ and $A_3$ 
do not include the $p_2$ and $p_3$ axes.}
\label{fig:intersectregionssphere2}
\end{figure}

\begin{lemma}\label{lem:propOAOBOC}The following holds:
\begin{enumerate}
\item[(i)] On $\OA \cup \OB$, the quantity $|p_2-p_3|$ is both
$\OO_c(\sqrt{|p_2|})$ and
$\OO_c(\sqrt{|p_3|})$.
\item[(ii)] On $\OB \cup \OC$, the quantity $|p_2-p_3|^{-1}$ is both
$\OO_c(|p_2|^{-1/2})$ and $\OO_c(|p_3|^{-1/2})$.
\end{enumerate}
\end{lemma}
\begin{proof}
Trivially, (i) holds because on $\OA \cup \OB$, we have the scaling
 $|p_3 - p_2| \lesssim \sqrt{|p_2+p_3|} \sim \sqrt{p_2} \sim \sqrt{p_3}$.
To obtain (ii), observe that either $p_2$
and $p_3$ have the same sign and by the definition of $\OB \cup \OC$,
$|p_2-p_3|\gtrsim  \sqrt{|p_2+p_3|} =  \sqrt{|p_2|+|p_3|}  \geq
\max(\sqrt{|p_2|}, \sqrt{|p_3|})$,
or they have a different sign and $|p_2-p_3| = |p_2| + |p_3| \geq \max(|p_2|,
|p_3|) \gtrsim \max(\sqrt{|p_2|}, \sqrt{|p_3|})$. In both cases, we have the
desired bound.
\myqed\end{proof}

We first work on $\OA\cup\OB$. In this region, $p_2$ and $p_3$ are close to
each other, and are both large in absolute value. It is then easy to find
a counterterm for 
${ {p_1W_\L w_\C    }/{p_2^{2}}}$
and
${ {p_4W_\R  w_\C}/{p_3^{2}}}$. Indeed, $W_\L$ and $W_\R$
oscillate very rapidly (the respective frequencies are approximately
$p_2/2\pi$ and $p_3/2\pi$), while $w_\C$ oscillates only ``moderately'', 
with frequency $(p_3-p_2)/2\pi$. One can then simply average
the rapidly oscillating part, and obtain
\begin{lemma}\label{lem:Hbloinresonances} Let $ \RAB  =  {p_1\Wl1 w_{\C}
  }/{p_2^{3}} - {  p_4\Wr1  w_{\C}}/{p_3^{3}}$.
Then,
\begin{equ}
L \RAB   ={\frac {p_1W_{\L} w_{\C}    }{p_2^{2}}} -{\frac {  p_4W_{\R}
w_{\C}}{p_3^{2}}} +  \OO_c(|p_2|^{-5/2}+|p_3|^{-5/2}) \qquad (\text{on }\OA
\cup
\OB)~.
\end{equ}
\end{lemma}
\begin{proof}
We have for the first term:
\begin{equs}
L  {\frac {p_1\Wl1 w_{\C}    }{p_2^{3}}} &= {\frac {p_1W_{\L} w_{\C}
}{p_2^{2}}} + \frac{p_1 \Wl1 w_{\C}' \cdot (p_3-p_2)}{p_2^3} + 
\OO_c(p_2^{-3}) \\
&= {\frac {p_1W_{\L} w_{\C}
}{p_2^{2}}} + 
\OO_c(|p_2|^{-5/2})~,
\end{equs}

where the last equality uses \lref{lem:propOAOBOC} (i).
A similar computation for the second term completes the proof.
\myqed\end{proof}

The counterterm $\RAB$ works well on $\OA \cup \OB$ because $|p_3 - p_2|$ 
is small compared to $p_2$ and $p_3$. 
We now have to find a counterterm $\RBC$ that works on $\OB \cup \OC$
and then patch the two counterterms 
together on $\OB$.
We state the properties of the counterterm $\RBC$ in the following
lemma, but postpone its construction to \sref{ss:decoupleddyn}.

\begin{lemma}\label{lem:Hbloindiagonale} There is a function
$\RBC  = \mathcal O_c(|p_2|^{-2} + |p_3|^{-2})$ defined on $\OB \cup \OC$
such that
\begin{equa}
L\RBC   &={\frac {p_1W_{\L} w_{\C}    }{p_2^{2}}} -{\frac {  p_4W_{\R}
w_{\C}}{p_3^{2}}}+   \OO_c(|p_2|^{-5/2}+|p_3|^{-5/2}) \qquad (\text{on }\OB
\cup
\OC)
\end{equa}
and
\begin{equ}[e:extderiveesr3]
\partial_{p_1}\RBC  = \OO_c (p_2^{-2}) \qquad \text{and} \qquad
\partial_{p_4}\RBC  = \OO_c (p_3^{-2})~.
\end{equ}
\end{lemma}

Assuming that \lref{lem:Hbloindiagonale} is proved, we
next join the two counterterms $\RAB$ and $\RBC$
by a smooth interpolation on $A_2$
in order to prove \pref{prop:technicalHb}.
\begin{proof}[Proof of \pref{prop:technicalHb}]
We introduce a smooth function $\rho : {\real}\cup\{-\infty, \infty\}\to [0,1]$
such
that $\rho(x) = 1$ when $|x|\leq 1$ and $\rho(x) = 0$ when $|x| \geq 2$.
We then consider the function
\begin{equ}[e:defchiargs]
	\rho\left(\frac{(p_3-p_2)^2}{ p_2+p_3}\right)~,
\end{equ}
which is well-defined and smooth on the set $A_* =  \OA \cup \OB \cup \OC = 
\{x\in \Omega: p_2\neq 0, p_3\neq 0\}$. Moreover,
it is equal to 1 on $\OA$, and 0 on $\OC$. We now omit the arguments
and simply write $\rho$ instead of \eref{e:defchiargs}.
Using \lref{lem:Hbloinresonances} and \lref{lem:Hbloindiagonale}, we obtain
\begin{equa}[eq:e:r1mr3Lx]
&L\left(\rho \RAB  +(1-\rho) \RBC \right) = \rho L\RAB   + (1-\rho)L\RBC  +
(\RAB  -\RBC )L\rho\\
&\qquad   ={\frac {p_1W_{\L} w_{\C}    }{p_2^{2}}} - {\frac {  p_4W_{\R}
w_{\C}}{p_3^{2}}}  + \OO_c (|p_2|^{-5/2}+|p_3|^{-5/2})+ (\RAB  -\RBC )L\rho~.
\end{equa}
Observe next that
\begin{equs}[e:bornechi]
L \rho= \rho'\cdot
&\left(2\frac{(p_3-p_2)}{ p_2+p_3}(w_{\L}-w_\R -2w_{\C}) + 
\frac{(p_3-p_2)^2}{ (p_2+p_3)^2}(w_{\L}+w_\R)\right)~.
\end{equs}
Since $\rho'$ has support in $\OB$, where $|p_3-p_2| \sim
|p_2+p_3|^\frac{1}{
2}  \sim |p_2|^{1/2}\sim |p_3|^{1/2} $, we see that $L\rho$
is simultaneously 
$\OO_c(|p_2+p_3|^{-1/2})$, $\OO_c(|p_2|^{-1/2})$ and $\OO_c(|p_3|^{-1/2})$.
But then, by \eref{eq:e:r1mr3Lx} and using that $\RAB  -\RBC =
\mathcal O_c(|p_2|^{-2} + |p_3|^{-2})$, we find
\begin{equ}\label{e:r1mr3Lx2}
L\left(\rho \RAB  +(1-\rho) \RBC \right)  ={\frac {p_1W_{\L} w_{\C}   
}{p_2^{2}}}
-{\frac {  p_4W_{\R}  w_{\C}}{p_3^{2}}}  +   \OO_c
(|p_2|^{-5/2}+|p_3|^{-5/2})~.
\end{equ}
We set now
\begin{equ}
\tH_c = \Hc3 + \rho \RAB   + (1-\rho)\RBC ~.
\end{equ}
From \eref{e:r1mr3Lx2} and \eref{e:round3},
we deduce immediately that \eref{e:Lhbn23} holds. Moreover,
\eref{e:Hctildep1p4derivatives} follows
from \eref{e:extderiveesr3}, the expressions
for $\Hc3$ and $\RAB$, and the fact that $\rho$ does not depend on $p_1$ and $p_4$.
This completes the proof of
\pref{prop:technicalHb}. 
\myqed\end{proof}

\mysubsection{Fully decoupled dynamics approximation}\label{ss:decoupleddyn}

We construct here the counterterm $\RBC$ of \lref{lem:Hbloindiagonale},
which eliminates the two resonant terms $-{ {p_1W_\L w_\C    }/{p_2^{2}}}$
and
${ {p_4W_\R  w_\C}/{p_3^{2}}}$ on $\OB\cup\OC$ when
both $|p_2|$ and $|p_3|$ are large. 
In this regime, 
all three interaction forces $w_\L, w_\C, w_\R$ oscillate rapidly
(since $|p_2|, |p_3-p_2|$ and $|p_3|$ are all large)
and we expect
the dynamics to be well approximated by the following {\em decoupled dynamics},
where all the interaction forces are removed.

\begin{definition}
We call {\em decoupled dynamics} the SDE
\begin{equs}[3][e:SDEdec]
\dd  q_i &= p_i\dt ~,\qquad &i=1,\dots,4~,\\
\dd p_b &=  - \gamma_b p_b \dt + \sqrt{2\gamma_b T_b} \dd B^b_t~,\qquad
&b=1,4~,\\
\dd p_j &=  0\dt~,\qquad &j=2,3~,
\end{equs}
with generator
\begin{equ}[eq:defLbar]
\barL  = \sum_{i =1}^4 p_i\partial_{q_i} +\sum_{b=1,4}\left(-\gamma_b p_b
\partial_{p_b} + \gamma_b T_b \partial^2_{p_b}\right)~,
\end{equ}
and denote by 
$\bar{\mathbb E}_{x}$ the corresponding expectation value
with initial condition $x \in \Omega$.
\end{definition}

We will construct two
functions $U_1$, $U_4$ such that $\bar L U_1 =p_1W_{\L} w_{\C}$
and $\bar LU_4 =  -p_4W_{\R}  w_{\C}$. Then, we will
introduce a change of variable $x\mapsto \bar x(x)$ such that
$\bar x$ approximately obeys the decoupled dynamics, so that
$L(U_1(\bar x)) \approx p_1W_{\L} w_{\C}  $ and 
$L(U_4(\bar x)) \approx  -p_4W_{\R}  w_{\C}$ in the regime of interest. 
Finally, we will show that the choice
$\RBC(x)  =  {U_1(\bar x)}/{p_2^2}  + {U_4(\bar x)}/{ p_3^2}$ satisfies
the conclusions of \lref{lem:Hbloindiagonale}.

The decoupled dynamics can be integrated explicitly for any
initial condition $x = (q_1, \dots, p_4)\in \Omega$.
For the outer rotors $b=1,4$, we have
\begin{equs}[3][e:dynamiqueLbarreexpliciteouter]
p_b(t) &= e^{-\gamma_b t}p_b +\sqrt{2\gamma_b T_b}\int_0^t e^{-(t-s)\gamma_b} \d
B_s^b~,\\
q_b(t) &= q_b +\frac{1-e^{-\gamma_b t}}{ \gamma_b}p_b  +\sqrt{2\gamma_b T_b}\int_0^t \Big(\int_0^s
e^{-(s-s')\gamma_b} \d B_{s'}^b\Big)\dd s~,
\end{equs}
and for the central ones ($j=2,3$) we simply have
\begin{equs}[3][e:dynamiqueLbarreexplicitecentre]
p_j(t) &= p_j~,\\
q_j(t) &= q_j + p_j t \pmod {2\pi}~,
\end{equs}
which is deterministic.
We decompose the variables between the central and external rotors as
\begin{equ}
	x = (x_e, x_c) \quad \text{with } x_e = (q_1, p_1, q_4, p_4) \quad \text{and } x_c = (q_2, p_2, q_3, p_3)~.
\end{equ}
Under the decoupled dynamics,
the two processes $x_e(t)$ and $x_c(t)$ are independent and $x_c(t)$ 
is deterministic. Moreover, under the decoupled dynamics, $x_e(t)$ has the generator 
$$
\barLO  = \sum_{b=1,4}(p_b\partial_{q_b} -\gamma_b p_b \partial_{p_b} + \gamma_b
T_b \partial_{p_b}^2 )~,
$$
and admits the invariant probability measure 
${\bar \pi}_e$ on  $(\torus \times \real)^2$
given by 
\begin{equ}
\d{\bar \pi}_e(x_e) = \frac{1}{Z} e^{- \frac{p_1^2}{ 2T_1}-
  \frac{p_4^2}{ 2 T_4}} \d q_1 \d p_1 \d q_4 \d p_4~,
\end{equ}
where $Z$ is a normalization constant (recall
that $T_1, T_4>0$ by assumption).

\begin{definition}
We denote by ${\SS}$ the set of functions $f\in \CC^\infty(\Omega,
\real)$ for which the norm
\begin{equ}[eq:normef]
\vertiii{f} = \sup_{x\in \Omega} \frac{| f(x) |}{1+p_1^{2}+p_4^{2}}
\end{equ}
is finite. We denote by ${\SS_0}$ the subspace of functions $f\in \SS$ for which
\begin{equ}
\int_{(\torus \times \real)^2} f(x_e, x_c)\d{\bar \pi}_e(x_e)  = 0 \qquad \text{for all
$x_c \in (\torus \times \real)^2$}~.
\end{equ}
\end{definition}

We will later consider $f = p_1W_{\L} w_{\C} $ and $f = -p_4W_{\R}  w_{\C}$,
which are manifestly in $\SS_0$.

\begin{lemma}\label{lem:esperancedecroitexp}
There are constants $C_*, c_*>0$ such that for all  $f\in {\SS}_0$, all $x\in \Omega$,
and all $t\ge0$,
\begin{equ}\label{e:expectationbardecroit}
\left|\bEE_{x} f(x(t))\right| \leq
C_*e^{-c_*t}\vertiii{f}\left(1+p_1^{2}+p_4^{2}\right)~.
\end{equ}
\end{lemma}
\begin{proof}
As mentioned, $x_e(t)$ and $x_c(t)$ are independent under the decoupled dynamics.
Introducing the expectation value $\Enull$  with respect to the process
$x_e(t)$ under the decoupled dynamics, we obtain that for any function $f$ on $\Omega$,
\begin{equ}[eq:decompexpectation]
\bEE_{x} f(x(t))  = \Enull_{x_e } f(x_e(t), x_c(t))~,
\end{equ}
where $x_c(t)$ is (deterministically) given by
 \eref{e:dynamiqueLbarreexplicitecentre}.
 
The process $x_e(t)$ under the decoupled dynamics
is exponentially ergodic, with
the unique invariant measure ${\bar \pi}_e$ defined above.
Indeed, one can check explicitly that this measure is invariant,
and introducing the Lyapunov function $V_e(x_e) =1+ p_1^{2} + p_4^{2}$,
we easily obtain that $\barLO V_e \leq c-cV_e$. It follows
from \cite[Theorem 6.1]{meyn_stability_1993}\footnote{One should also check
that there is a skeleton with respect to which every compact set
is petite. This is obvious, but can be proved with methods
similar to those of §5 in \cite{cuneo_nonequilibrium_2014}.} that 
there are two constants $C_*, c_*>0$ such that for
any function $g:(\torus\times \real)^2 \to \real$ such that
$g/V_e$ is bounded,
\begin{equ}[eq:supxg]
\sup_{x_e}\frac{|\Enull_{x_e}
g(x_e(t)) - {\bar \pi}_e(g)|}{ 1+  p_1^{2} + p_4^{2}}\leq
C_*e^{-c_*t} \sup_{x_e} \frac{|g(x_e) - {\bar \pi}_e(g)
|}{ 1 +  p_1^{2} + p_4^{2}}~.
\end{equ}
Let now $f\in \SS_0$. For any fixed $v \in (\torus \times \real)^2$, we apply \eref{eq:supxg}
to the function $g_v(x_e) = f(x_e, v)$. Since $f\in \SS_0$,
we have ${\bar \pi}_e(g_v) = 0$. Therefore, for any $t\geq 0$, 
\begin{equ}\label{e:fexp}
\sup_{x_e}\frac{|\Enull_{x_e} f(x_e(t), v)|}{ 1+p_1^{2} +p_4^{2}}\leq C_*e^{-c_*t}
\sup_{x_e} \frac{|f(x_e,v)|}{1+p_1^{2}+p_4^{2}} \leq C_*e^{-c_*t}\vertiii{f}~.
\end{equ}
This holds for all $v$, and in particular for $v = x_c(t)$. Therefore, by
\eref{eq:decompexpectation}, we have the desired result.
\myqed\end{proof}

The next proposition constructs a right inverse of $\barL$ on $\SS_0$
\cite{MR1872736,MR1988467,MR2135314}.
We use here the notation
\begin{equ}[eq:notationx1x8]
	x=(x_1, \dots, x_8) = (q_1, \dots, q_4, p_1, \dots, p_4)~.
\end{equ}

\begin{proposition}\label{prop:rbarref}
Let $f\in  \SS_0$ be a function such that for all multi-indices
$\underline a$, we have $\partial^{\underline a}f\in  \SS_0$,
and let
\begin{equ}[eq:defKbar](\barR  f)(x)=-\int_0^\infty  \bar{\mathbb E}_{x} f(x(t)) \,
  \d t~.
\end{equ}
Then:
\begin{enumerate}
\item[(i)] $\barR  f$ and its derivatives of all orders are in $\SS$.
\item[(ii)] We have
\begin{equ}
\barL  \barR  f = f~.
\end{equ}
\end{enumerate}
\end{proposition}
\begin{proof}
By \lref{lem:esperancedecroitexp}, the integral \eref{eq:defKbar} converges
absolutely for all $x$ and we have $\barR f \in \SS$.
We now prove the result about the derivatives. By
\eref{e:dynamiqueLbarreexpliciteouter}
and \eref{e:dynamiqueLbarreexplicitecentre}, we can write
\begin{equs}[eq:linearitydec]
	\frac{\partial x_i(t)}{\partial x_j} = h_{ij}(t)~,
\end{equs}
where the $h_{ij}$ are deterministic functions of $t$ only that grow
at most linearly (namely $0$, $1$,
$e^{-\gamma_b t}$, $(1-e^{-\gamma_b t})/{ \gamma_b}$ and $t$).
We then have
\begin{equ}
\frac{\partial }{\partial x_j} \left(\bar{\mathbb
E}_{x} f(x(t))\right) = \sum_{i=1}^8 h_{ij}(t)\bar{\mathbb
E}_{x}[(\partial_i f)(x(t))]~.
\end{equ}
For the derivatives of order $n$, we find by induction
\begin{equ}[eq:deriveeKf]
\partial_{j_1, \dots, j_n}\left(\bar{\mathbb
E}_{x} f(x(t))\right)  = \sum_{i_1, \dots, i_n}
\Big(\prod_{k=1}^{n} h_{i_k j_k}(t)\Big)\bar{\mathbb
E}_{x}[(\partial_{i_1, \dots, i_n} f)(x(t))]~,
\end{equ} 
where the sum is taken over all $(i_1, \dots, i_n)\in \{1, 2, \dots, 8\}^n$.
Since by assumption $\partial_{i_1, \dots, i_n} f\in \SS_0$, we have
by \lref{lem:esperancedecroitexp} that
$$\left|\bar{\mathbb
E}_{x}[(\partial_{i_1, \dots, i_n} f)(x(t))]\right| \leq c e^{-c_*t}
\left(1+p_1^{2}+p_4^{2}\right)~.$$
But then, by \eref{eq:deriveeKf}, we have
\begin{equs}
\left|\partial_{i_1, \dots, i_n} \barR f(x)\right|& = \left| \int_{0}^\infty
\partial_{i_1, \dots, i_n}\left(\bar{\mathbb
E}_{x} f(x(t))\right) \d t \right|\\
& \leq c \left(1+p_1^{2}+p_4^{2}\right) \sum_{i_1, \dots, i_n} \left|\int_0^\infty 
\Big(\prod_{k=1}^{n} h_{i_k j_k}(t)\Big) e^{-c_*t} \dt\right|~.
\end{equs} 
Since the $h_{ij}$ grow at most linearly, the time-integrals in the right-hand
side converge. Therefore, $\barR f$ is $\CC^\infty$ and (i) holds.

For the second statement, we observe that
\begin{equ}
 \barL  \barR f = - \int_{0}^\infty \barL  \bar{\mathbb E}_{x} f(x(t))\d t =-
\int_{0}^\infty \frac{\dd}{\dt} \bar{\mathbb E}_{x}f(x(t)) \d t =\bar{\mathbb E}_{x}
f(x(0)) = f(x)~,
\end{equ}
where we have used that $\lim_{t\to\infty}  \mathbb E_x f(x(t)) = 0$ by
\eref{e:fexp}.
\myqed\end{proof}

\begin{remark}\label{rem:paspinningettau} 
The proof of \pref{prop:rbarref}, 
and in particular \eref{eq:linearitydec}, relies on the linear nature of
the decoupled dynamics.
If we add constant forces $\tau_1$ and $\tau_4$
at the ends of the chain (as in \cite{cuneo_nonequilibrium_2014}),
the method above applies with little
modification, and with the replacements $p_b \to p_b-\tau_b/\gamma_b$,
$b=1,4$, in the invariant measure $\bar \pi_e$.
However, if we add pinning potentials of the kind $U(q_i)$,
the decoupled dynamics cannot be solved explicitly, and we
do not have \eref{eq:linearitydec} for some deterministic functions
$h_{ij}(t)$. Although we believe there exists an analog of \pref{prop:rbarref}
in that case, we are currently unable to provide it.
The situation is even worse in the simultaneous presence of constant forces
and pinning potentials. In that case, the expression of
${\bar \pi}_e$ is not known \cite{gallavotti_nonequilibrium_2013}, which makes
it difficult to decide whether a given function is in $\SS_0$.
(Of course, although there is no difficulty there, 
the averaging of $p_2$, $p_3$
and $H_c$ also needs to be adapted to accommodate for such
modifications of the model.)
\end{remark}

We now have an inverse of $\bar L$ on a given class of functions.
We next use it to find an approximate inverse of $L$.
The key is to introduce a change of variables $\bar x = (\bar q_1,
\bar p_1, \dots, \bar  q_4, \bar  p_4)$ such that for nice
enough functions $f$, it holds that
$L(f(\bar x)) \approx (\barL f)(\bar x)$ in the regime of interest.
Here and in the sequel,
it is always understood that $\bar x$ is viewed as a function of $x$.
We compare the actions of $L$ and $\bar L$
in \lref{lem:lbar}. We state this lemma with the notation \eref{eq:notationx1x8},
and write generically
\begin{equ}[eq:defLLbar]
L =  \sum_i( b_i(x)\partial_i + \sigma_i\partial_i^2) \qquad \text{and}\quad 
\bar L = \sum_i (  \bar b_i(x)\partial_i + \sigma_i\partial_i^2)~.
\end{equ}
In our case, only $\sigma_5$ and $\sigma_8$, which correspond to
the variables $p_1$ and $p_4$, are non-zero.

\begin{lemma}\label{lem:lbar}
Consider a change of coordinates
$x\mapsto \bar{x}(x) = x + s(x)$, defined on some set $\Omega_0\subset \Omega$.
Assume that for all $j$, 
$$
L (\bar x_j) = \bar b_j(\bar x) + \epsilon_j(x)
$$
for some $\varepsilon_j$. Then, for any smooth function $h$, we
have
for all $x\in \Omega_0$ that  
\begin{equ}
L(h(\bar x)) = (\bar L h)(\bar{x}) + \zeta(x)~,
\end{equ}
where
\begin{equ}\label{eq:zetax}
\zeta(x) =  \sum_j (\partial_j h)(\bar{x}) \epsilon_j(x) + 2\sum_{i,k}\sigma_i
(\partial_{ik}h)(\bar{x})\partial_i s_k(x) + \sum_{i,j,k}\sigma_i
(\partial_{jk}
h)(\bar{x})\partial_i s_j(x)\partial_i s_k(x)~.  
\end{equ}
\end{lemma}

\begin{proof}We do the computation for the case of just one variable $x\in \real$.
Let $g(x) = \bar x(x) = x + s(x)$.
From the definition of $L$ and $\bar L$, and since by
assumption $Lg=\bar b \circ g + \epsilon $,
we find
\begin{equa}
L(h\circ g) &=  (h'\circ g)\cdot  L g +  \sigma \cdot ( h''\circ g )\cdot
{g'}^2\\
& =  (h'\circ g)  \cdot (\bar b \circ g + \epsilon)  +  \sigma\cdot (
h''\circ g)\cdot {g'}^2\\
& =  (\bar Lh)\circ g  +  ( h'\circ g )\cdot  \epsilon  + \sigma \cdot 
(h''\circ g)\cdot ( {g'}^2 - 1) \\
& =  (\bar Lh)\circ g  +  ( h'\circ g )\cdot  \epsilon  + \sigma \cdot 
(h''\circ g)\cdot ( 2s' + s'^2)~.
\end{equa}
The desired result follows from generalizing to the multivariate case.
\myqed\end{proof}

We consider now the following change of variables defined
on $\OB\cup\OC $: 
\begin{equa}[2][e:bar]
\bar q_1 = q_1~,&\qquad\quad \bar p_1 =p_1-{\frac {W_{\L} \left( q_2-q_1\right)
}{p_2}} = p_1 + \OO_c(p_2^{-1})~,\\
\bar q_2 = q_2~,&\qquad\quad \bar p_2 =p_2+{\frac {W_{\L} \left( q_2-q_1\right)
}{p_2}}-{\frac {W_\C
\left( q_3-q_2\right) }{p_3-p_2}} = p_2 + \OO_c(|p_2|^{-1/2})~,
\end{equa}
with analogous expressions for the indices 3, 4. Here, we have used
\lref{lem:propOAOBOC} (ii) to replace $W_\L/(p_3-p_2)$
with $\OO_c(|p_2|^{-1/2})$.
Straightforward computations show that, on $\OB\cup\OC $,
\begin{equa}[1][e:lpbar]
L(\bar p_1)  &= -\gamma_1 p_1+ \OO_c(p_2^{-1}) = - \gamma_1 \bar p_1 +
\OO_c(p_2^{-1})~,\\
 L(\bar q_1) & = p_1 =  \bar p_1 +\OO_c(p_2^{-1})~,\\
L(\bar p_2)  &= \OO_c(|p_2|^{-1}+(p_2-p_3)^{-2})= \OO_c(p_2^{-1})~,\\
L(\bar q_2)  &=  p_2 =  \bar p_2 + \OO_c(|p_2|^{-1/2})~,
\end{equa}
with similar expressions for the indices $3,4$ (we have again used \lref{lem:propOAOBOC} (ii)).

While one could choose a more refined change of variables
by going to higher orders, the change \eref{e:bar} is good enough
for our purpose.

\begin{lemma}\label{lem:passageOOxi}
Let $f\in \SS$. Then $f$ is $\OO_c(1)$. Moreover,
given any function $\xi: \Omega \to [0,1]$, we have that
$f\bigl(x + \xi(x)(\bar x -x)\bigr)
= \OO_c(1)$ on $\OB\cup\OC$. In particular, $f(\bar x) = \OO_c(1)$
on $\OB\cup\OC$.
\end{lemma}
\begin{proof} By assumption,
$|f(x)| \leq \vertiii{f}(1+p_1^2+p_4^2)$ on $\Omega$,
so that immediately $f = \OO_c(1)$. Moreover, 
$f\bigl(x + \xi(x)(\bar x -x)\bigr)$ is well-defined on
$\OB\cup\OC$, and
\begin{equs}
\left|f\bigl(x + \xi(x)(\bar x -x)\bigr)\right| &\leq \vertiii{f}\left(1+ \left(p_1 -
\xi(x)\frac{W_{\L}}{p_2}\right)^2 +\left(p_4
- \xi(x)\frac{W_{\R}}{p_3} \right)^2 \right)~,
\end{equs}
which is indeed a $\OO_c(1)$ on this set. The claim about $f(\bar x)$
follows from the choice $\xi \equiv 1$.
\myqed\end{proof}

\begin{proposition}\label{prop:changevarours}Let $f$ satisfy
the assumptions of \pref{prop:rbarref}, and consider the change of coordinates 
\eref{e:bar}. Let $h = \bar Kf$. Then, on the set $\OB \cup \OC$ (meaning
that we take $x\in \OB \cup \OC$, and not necessarily $\bar x\in \OB\cup \OC$), we have
$h(\bar x) = \OO_c(1)$ and
\begin{equs}
L (h(\bar x)) &= f( x) +  \OO_c(|p_2|^{-1/2} + |p_3|^{-1/2})~.
\end{equs}
\end{proposition}
\begin{proof}
We use again the notations $x=(x_1, \dots, x_8) = (q_1, \dots, q_4, p_1, \dots,
p_4)$ and \eref{eq:defLLbar}.
We apply \lref{lem:lbar} with the coordinate change $\bar x = x + s(x)$ defined
by \eref{e:bar}. Then, the
$s_j$ are given by \eref{e:bar}, and the 
$\varepsilon_j$ are given by \eref{e:lpbar}. Observe then
that on $\OB\cup\OC$, all the $s_j$ and $\varepsilon_j$
and are at most $\OO_c(|p_2|^{-1/2})$ or $\OO_c(|p_3|^{-1/2})$.
The only non-zero $\sigma_i$ are
$\sigma_{5} = \gamma_1 T_1$ and $\sigma_{8} = \gamma_4 T_4$.
Moreover, $\partial_{x_5}s_j = \partial_{p_1} s_j =0$ for
all $j\in \{1, 2, \dots, 8\}$, and similarly
$\partial_{x_8}s_j = \partial_{p_4} s_j =0$.
Therefore, from \eref{eq:zetax} we are left with
$\zeta(x) = \sum_j (\partial_j h)(\bar{x}) \epsilon_j(x)$.
We now apply this to the function $h = \bar K f$. 
By \pref{prop:rbarref}, we have $\bar L h = f$, so that
\begin{equa}[e:LU1]
L(h(\bar{x})) & =  f(\bar x) +  \sum_j(\partial_j h)(\bar{x}) \epsilon_j(x)~.
\end{equa}
To obtain the desired results, it remains to make the following two
observations.
First, by the mean value theorem, there is for each $x$ some
$\xi(x) \in [0,1]$ such that on $\OB\cup\OC$,
\begin{equ}[eq:fbarxmfbar]
f(\bar x)-f(x) = \sum_j  s_j(x) (\partial_j f)(x+\xi(x)s(x)) =
\OO_c(|p_2|^{-1/2}+|p_3|^{-1/2})~,
\end{equ}
where we have applied \lref{lem:passageOOxi} to $\partial_jf$,
which is in $\SS$ by assumption.
Secondly, using \lref{lem:passageOOxi}
and the fact that $\partial_j h\in \SS$ by \pref{prop:rbarref}, we find 
\begin{equa}
\sum_j(\partial_j h)(\bar{x}) \epsilon_j(x) =\OO_c(|p_2|^{-1/2}+|p_3|^{-1/2})~,
\end{equa}
which, together with \eref{e:LU1} and \eref{eq:fbarxmfbar}, completes the
proof.
\myqed\end{proof}

We are now ready for the

\begin{proof}[Proof of \lref{lem:Hbloindiagonale}]
Let
\begin{equa}
U_1(q_1, \dots, q_3, p_1, \dots, p_3) &= \barR (p_1W_{\L}(q_2-q_1)
w_{\C}(q_3-q_2))~,\\\
U_4(q_2, \dots, q_4, p_2, \dots, p_4) &= \barR (-p_4W_{\R}(q_3-q_4)
w_{\C}(q_3-q_2))~,
\end{equa}
and
$$
\RBC(x)  =  \frac{U_1(\bar x)}{
p_2^2}  + \frac{U_4(\bar x)}{ p_3^2}~.
$$
That $U_1$ depends only on $(q_1, \dots, q_3, p_1, \dots, p_3)$ follows from
the independence of the four rotors under the decoupled dynamics.
Similarly for $U_4$.
It is easy to check that $f=p_1W_{\L}w_{\C}$ satisfies the assumptions of
\pref{prop:rbarref}: Since $\avg{f}_{1} = 0$, we also have
$\avg{\partial^{\underline a } f}_{1} = 0$ for each multi-index $\underline a$.
From this it follows that
${\bar \pi}_e(f) = 0$
and that ${\bar \pi}_e(\partial^{\underline a } f) = 0$, since
${\bar \pi}_e$ is uniform with respect to $q_1$. Since no powers of
$p_1$
or $p_4$ appear upon differentiation, we indeed obtain that $f$ and all its
derivatives are in $\SS_0$. A similar argument applies to
$f=-p_4W_{\L}w_{\C}$.
Therefore, applying \pref{prop:changevarours}, we find that on
the set $\OB\cup\OC$, the functions
$U_1(\bar x)$ and $U_4(\bar x)$ are $\OO_c(1)$,
and that
\begin{equs}[e:Lu1barx]
L(U_1(\bar{x})) &  = p_1W_{\L}w_{\C} + \OO_c(|p_2|^{-1/2}+|p_3|^{-1/2})~,\\
L(U_4(\bar{x})) & =  - p_4W_{\R} w_{\C} + \OO_c(|p_2|^{-1/2}+|p_3|^{-1/2})~.
\end{equs}
In \eref{e:Lu1barx}, the arguments of $W_\L, W_\R$ and $W_\C$ are 
indeed $x$ and not $\bar x$.
Finally, we have
\begin{equ}\label{e:lr23}
L\RBC  =  \frac{L(U_1(\bar x))}{  p_2^2}  + \frac{L(U_4(\bar x))}{ p_3^2} +
\OO_c(p_2^{-3})+ \OO_c(p_3^{-3})~.
\end{equ}
The main assertion of the lemma then follows from this, \eref{e:Lu1barx}, and
\lref{lem:repartitionpowers}. The assertion \eref{e:extderiveesr3} follows from
the definition of $R_{23}$ and the following observation:
using the explicit expression for $\bar x$, \pref{prop:rbarref} (i)
and \lref{lem:passageOOxi}, we obtain
$\partial_{p_1}(U_1(\bar x)) = (\partial_{p_1}U_1)(\bar x) = \OO_c(1)$, and
$\partial_{p_4}(U_1(\bar x)) = (\partial_{p_4}U_1)(\bar x) = 0$ (and similarly for $U_4$).
\myqed\end{proof}

\begin{remark}\label{rem:Teq0Harris} The construction above relies
on the strict positivity of the temperatures (which we assume throughout).
Nonetheless, it can be adapted to
the case $T_1 = T_4 = 0$. In this case,
the external rotors are not ergodic under the decoupled dynamics:
they deterministically slow down and asymptotically
reach a given position that depends on the initial condition.
Therefore, the conclusion of \lref{lem:esperancedecroitexp} does
not hold.
However, the counterterm $\RBC$ that we obtained still produces the desired
effect. Indeed, at zero temperature, the definition of $U_1$
becomes
$$
U_1(x) = - \int_0^\infty p_1(t)W_{\L}(q_2(t)-q_1(t))w_{\C}(q_3(t)-q_2(t))~,
$$
where $x(t)$ is the deterministic solution given in \eref{e:dynamiqueLbarreexpliciteouter}
and \eref{e:dynamiqueLbarreexplicitecentre}
with initial condition $x$ and $T_1=T_4=0$. Since $p_1(t)$ decreases
exponentially
fast and $W_\L w_\C$ is bounded, this integral still converges. A similar
argument applies to $U_4$.
\end{remark}

\section{Constructing a global Lyapunov function}\label{s:constrlyapunov}
We construct here the Lyapunov function
of \tref{prop:Lyapunov}. We start by fixing
the parameters defining the sets $\Omega_2, \Omega_3, \Omega_c$
and the functions $V_2, V_3, V_c$.

We assume throughout this section that $\theta$
is fixed and satisfies
\begin{equ}[eq:conditionbeta]
0 <\theta < \min\left(\frac 1{T_1},\frac 1{T_4}\right).
\end{equ}

This condition is necessary to apply \pref{c:Lelambdap2tilde}
and \pref{cor:borneOmegac}. In addition, it guarantees
that when $p_1^2+p_4^2$
is large, $\exp(\theta H)$  decreases very fast:
\begin{lemma}\label{l:LebetaH} There are
constants $\CstEbha, \CstEbhb>0$ such that
\begin{equ}\Label{e:LbetaHseul}
Le^{\theta H} \leq  (\CstEbha - \CstEbhb(p_1^2 + p_4^2))e^{\theta H}~.
\end{equ}
\end{lemma}
\begin{proof} Since $
Le^{\theta H} = \sum_{b=1,4}  \left(-\gamma_b\theta (1-\theta T_b)
p_b^{2}+\gamma_b\theta
T_b\right)e^{\theta H}
,$ the result follows from the condition on $\theta$.
\myqed\end{proof}

We next choose the constants $k$, $\ell$, $a$, $m$, and finally $R$.
First, we fix $k$ large enough, and require a lower bound $R_0$ on $R$
so that the conclusions of \pref{c:Lelambdap2tilde} hold on
$\Omega_j(k, R)$, $j=2,3$.
We then fix the parameters $a$ (appearing in $V_2$, $V_3$) 
and $\ell$ such that
\begin{equ}[eq:condition2lak]
\frac 2{\ell} < a < \frac 2{k}~.
\end{equ}
As a consequence, $\Omega_c(\ell, m, R)$ now depends only on $m$ and $R$,
which we fix large enough so that \pref{cor:borneOmegac} applies,
and so that $m>\ell$ and $R\geq R_0$.

This choice satisfies the condition 
$1\leq  k < \ell < m$ imposed in \eref{eq:conditionklm}. This ensures that
the sets $\Omega_j$ ($j=2,3$) and $\Omega_c$ have ``large'' intersections,
and that they indeed look as shown in \fref{fig:intersectregionssphere}
and \fref{fig:intersectregionsp2p3}.
Moreover, condition \eref{eq:condition2lak} ensures that for large $|p_j|$,
$j=2,3$,
$$
|p_j|^{2/\ell } \ll |p_j|^a \ll |p_j|^{2/k}~,
$$
which will be crucial.

We next introduce smooth cutoff functions for the sets
$\Omega_2, \Omega_3, \Omega_c$. For this, we consider
for each set a thin ``boundary layer'' included in the set itself.

\begin{definition}Let $\Pnull$ be a subset of the momentum space $\real^4$.
We define $\BB(\Pnull) = \{p\in \Pnull:
{\rm dist}(p,\Pnull^c)<1\}$.
\end{definition}

\begin{lemma}\label{l:setscutoff}Let $\Pnull \subset \real^4$.
Then, there is a smooth function $\psi:\real^4 \to[0,1]$ with the following
properties. First,
$\psi(p)=1$ on
$\Pnull\setminus\mathcal B(\Pnull)$ and $\psi(p)=0$ on 
$\Pnull^c$, with some interpolation on $\mathcal B(\Pnull)$.
Secondly, $\partial^{\underline a} \psi$
is bounded on $\mathbb R^4$ for each multi-index $\underline a$.
\end{lemma}
\begin{proof}
Such a function is obtained 
by appropriately regularizing the characteristic function of the set
$\{p\in \Pnull: {\rm dist}(p, \Pnull^c) > 1/2\} \subset \real^4$. 
\myqed\end{proof}

Since the definition of sets $\Omega_c$ and $\Omega_j$, $j=2,3$,
involves only the momenta, we can write $\Omega_c = \torus^4 \times \Pnull_c$
and $\Omega_j = \torus^4 \times \Pnull_j$ for some sets $\Pnull_c, \Pnull_j
\subset \real^4$. We  apply \lref{l:setscutoff}
to $\Pnull_c,\Pnull_2$ and $\Pnull_3$, and denote by  $\psi_c, \psi_2$, and $\psi_3$
the functions obtained. We introduce also the sets
\begin{equ}
\BB(\Omega_c) =	\torus^4 \times \BB(\Pnull_c), \quad  \BB(\Omega_2) =	\torus^4 \times \BB(\Pnull_2),
\quad  
\BB(\Omega_3) =	\torus^4 \times \BB(\Pnull_3)~.
\end{equ}

Obviously, $\BB(\Omega_c) \subset \Omega_c$ and $\BB(\Omega_j) \subset \Omega_j$.

\begin{proof}[Proof of \tref{prop:Lyapunov}]
We show that the Lyapunov function
$$
V = 1+ e^{\theta H} + \sum_{j=2,3}\psi_j(p)V_j + \Y\psi_c(p)V_c
$$
has the necessary properties, provided that the constant $\Y$ is large
enough. We start by proving \eref{e:ineqvn}. From \eref{e:ptij} and
\eref{e:bornesVcOmegac},
we immediately obtain the bound
\begin{equ}\label{e:ineqvnre}
1+e^{\theta H}\leq V\leq c(\psi_2 e^{|p_2|^a}+\psi_3 e^{|p_3|^a} +
\psi_c\cdot(p_2^2+p_3^2))e^{\theta H}~,
\end{equ}
which is slightly sharper than \eref{e:ineqvn}. 
We next turn to the bound on $LV$. We introduce
the set
\begin{equ}[eq:defG]
G=\{x\in \Omega : p_1^2+p_4^2 < (1+\CstEbha)/\CstEbhb \}~,
\end{equ} 
with $\CstEbha$, $\CstEbhb$ as in \lref{l:LebetaH}, so that\begin{equ}[eq:LethetaindG]
Le^{\theta H} \leq -e^{\theta H} + (1+\CstEbha)\ind_{G}e^{\theta H} \leq-e^{\theta
H} + \CstGa \ind_{G} e^{\frac \theta 2 (p_2^2 + p_3^2)}
\end{equ}
for some $\CstGa > 0$, where we have used that
$H \leq c+ \frac{p_2^2}2 + \frac{p_3^2}2$ on $G$.
Moreover, observe that for $j=2,3$, there is a polynomial $z_j(p)$ such that
\begin{equs}[eq:LpsijVjpolyn]
L(\psi_jV_j)& = \psi_j LV_j + V_j L\psi_j + 2\sum_{b=1,4} \gamma_b T_b
(\partial_{p_b}\psi_j)(\partial_{p_b}V_j)\\
& \leq -\CstOmegaid \psi_j p_j^{-2} e^{|p_j|^a + \frac \theta 2  p_j^2} +  V_j
L\psi_j + 2\sum_{b=1,4} \gamma_b T_b
(\partial_{p_b}\psi_j)(\partial_{p_b}V_j)\\
& \leq  (-\ind_{\Omega_j}\CstOmegaid p_j^{-2} + \ind_{\mathcal B(\Omega_j)}
z_j(p))e^{|p_j|^a + \frac \theta 2  p_j^2}~,
\end{equs}
where the first inequality follows from \eref{e:Lfp2tilde} and the second
inequality holds
because the derivatives of $\psi_j(p)$ have support on $\BB(\Omega_j)$,
because 
$|\psi_j-\ind_{\Omega_j}| \leq \ind_{\BB(\Omega_j)}$, and because of
\eref{e:ptij}.
Similarly, using \pref{cor:borneOmegac}, we obtain a polynomial $z_c(p)$ such
that on $\Omega$
\begin{equs}[eq:LpsicVcpoly]
L(\psi_c(p)V_c) & \leq (-\ind_{\Omega_c}\CstOmegacd + \ind_{\mathcal
B(\Omega_c)}z_c(p) )e^{\frac \theta 2 (p_2^2 + p_3^2)}~.
\end{equs}
Combining \eref{eq:LethetaindG}, \eref{eq:LpsijVjpolyn} and
\eref{eq:LpsicVcpoly}, we find
\begin{equs}[eq:premierLV]
LV &\leq  -e^{\theta H}  -\sum_{j=2,3}\ind_{\Omega_j}\CstOmegaid
p_j^{-2}e^{|p_j|^a + \frac \theta 2  p_j^2}-\ind_{\Omega_c}\Y\CstOmegacd
e^{\frac \theta 2 (p_2^2 + p_3^2)} \\
&\quad +  \CstGa\ind_{G}e^{\frac \theta 2 (p_2^2 + p_3^2)} +
\sum_{j=2,3}\ind_{\mathcal B(\Omega_j)} z_j(p)e^{|p_j|^a + \frac \theta 2 
p_j^2}
+ \Y \ind_{\mathcal B(\Omega_c)}z_c(p)e^{\frac \theta 2 (p_2^2 + p_3^2)}~.
\end{equs}

The first line contains the ``good'' terms. We next show that these terms
dominate the others. Let $\varepsilon > 0$.
We claim that there is a (large) compact set $K$ (which depends on $\varepsilon$) such that
\begin{equ}[eq:indOmegaJleq]
\ind_{\mathcal B(\Omega_j)}z_j(p)e^{|p_j|^a + \frac \theta 2  p_j^2}
\leq \epsilon e^{\theta H} + c \ind_K~,\quad j=2,3~,
\end{equ}
\begin{equ}[eq:indOmegacleq]
\ind_{\mathcal B(\Omega_c)} z_c(p)e^{\frac \theta 2 (p_2^2 + p_3^2)}
 \leq \epsilon
e^{\theta H} + \epsilon \sum_{j=2,3}\ind_{\Omega_j} p_j^{-2}e^{|p_j|^a + \frac
\theta 2  p_j^2} + c\ind_K~,
\end{equ}
\begin{equ}[eq:indG]
\ind_G e^{\frac \theta 2 (p_2^2 + p_3^2)} \leq
\ind_{\Omega_c}e^{\frac \theta 2 (p_2^2 + p_3^2)} + \epsilon\sum_{j=2,3}\ind_{\Omega_j} p_j^{-2}e^{|p_j|^a + \frac \theta 2  p_j^2}
 +c\ind_K~.
\end{equ}
We prove these bounds one by one.
\begin{itemize}
\item Proof of \eref{eq:indOmegaJleq}. 
We prove the bound for $j=2$. First observe that 
\begin{equ}[eq:ethetagc]
z_2(p)e^{|p_2|^a + \frac \theta 2  p_2^2 -  \theta H} < c z_2(p) e^{|p_2|^a -\frac {\theta}{2}(p_1^2  + p_3^2 + p_4^2) }~.
\end{equ}
By the definition of $\Omega_2$, when $\|p\|\to
\infty$ in $\mathcal B(\Omega_2)$, we find $p_1^2  + p_3^2 + p_4^2 \sim |p_2|^{2/k} \gg |p_2|^a$
(recalling that $2/k > a$). Thus, the right-hand side of \eref{eq:ethetagc}
vanishes in this limit, since $z_2$ is only a polynomial. This implies 
\eref{eq:indOmegaJleq} if $K$ is large enough.

\item Proof of \eref{eq:indOmegacleq}. 
By inspection of the definition \eref{e:omegac} of $\Omega_c$, there are
three regions $B(\Omega_c)_i$, $i=1,2,3$, such that if the compact set $K$
is large enough,
$$
\mathcal B(\Omega_c) \subset K\cup \mathcal B(\Omega_c)_1\cup\mathcal B(\Omega_c)_2 \cup\mathcal B(\Omega_c)_3~,
$$
where $\mathcal B(\Omega_c)_1$ is such that $p_2^2 + p_3^2 \sim
(p_1^2+p_4^2)^m$, where $\mathcal B(\Omega_c)_2 \subset \Omega_2$
is such that $p_3^{2\ell} \sim  p_2^2$,
and where $\mathcal B(\Omega_c)_3 \subset \Omega_3$ is such that $p_2^{2\ell} \sim  p_3^2$
(see \fref{fig:intersectregionssphere} and \fref{fig:intersectregionsp2p3}).

Since $ z_c(p)e^{\frac \theta 2 (p_2^2 + p_3^2)}$ is bounded on compact sets,
\eref{eq:indOmegacleq} trivially holds on $K$.
We next turn to $\mathcal B(\Omega_c)_1$. We have 
$$
z_c(p)e^{\frac \theta 2 (p_2^2 +p_3^2) - \theta H} < c z_c(p)e^{-\frac \theta 2 (p_1^2 +p_4^2) }~.
$$
The right-hand side vanishes when $\|p\| \to \infty$ in $\mathcal B(\Omega_c)_1$, 
and thus by enlarging $K$ if necessary, we find $ z_c(p)e^{\frac \theta 2 (p_2^2 + p_3^2)}
 \leq \epsilon
e^{\theta H}+c\ind_K$ on  $\mathcal B(\Omega_c)_1$, which implies \eref{eq:indOmegacleq}.

Now, consider $\mathcal B(\Omega_c)_2$. We have
$$
\frac{z_c(p)e^{\frac \theta 2 (p_2^2 +p_3^2)}}{p_2^{-2}e^{|p_2|^a + \frac \theta 2  p_2^2}}
 = z_c(p)p_2^2e^{\frac \theta 2 p_3^2 - |p_2|^a }~.
$$
As $\|p\| \to \infty$ in $\mathcal B(\Omega_c)_2$, the right-hand side vanishes,
since $p_3^{2} \sim  |p_2|^{2/\ell}$ and $a>2/\ell$. Therefore,
$ z_c(p)e^{\frac \theta 2 (p_2^2 + p_3^2)}
 \leq\epsilon p_2^{-2}e^{|p_2|^a + \frac
\theta 2  p_2^2} + c\ind_K$ on  $\mathcal B(\Omega_c)_2$ for large enough $K$,
and thus \eref{eq:indOmegacleq}
holds on $\mathcal B(\Omega_c)_2$ since $\mathcal B(\Omega_c)_2 \subset \Omega_2$.

A similar argument applies for $\mathcal B(\Omega_c)_3$, which completes the
proof of \eref{eq:indOmegacleq}.

\item Proof of \eref{eq:indG}. Observe that
for $K$ large enough, the set $G$ defined in \eref{eq:defG} verifies
$$
G \subset K \cup \Omega_2\cup\Omega_c \cup\Omega_3~.
$$
On $K$ and $\Omega_c$, \eref{eq:indG} holds trivially.
On $G \cap \Omega_2 \setminus \Omega_c$ and for large enough $\|p\|$, we have
$p_3^2 \leq |p_2|^{2/\ell}$ (otherwise we would have $x\in \Omega_c$), and therefore
$$
\frac{e^{\frac \theta 2 (p_2^2 +p_3^2)}}{p_2^{-2}e^{|p_2|^a + \frac \theta 2  p_2^2}}
 = p_2^2e^{\frac \theta 2 p_3^2 - |p_2|^a } \leq   p_2^2e^{\frac \theta 2 |p_2|^{2/\ell}- |p_2|^a } ~.
$$
Since $a>2/\ell$,  and by enlarging $K$ if necessary, we have $e^{\frac \theta 2 (p_2^2 + p_3^2)} \leq \epsilon p_2^{-2}e^{|p_2|^a + \frac \theta 2  p_2^2}
 +c\ind_K$ on $G \cap \Omega_2 \setminus \Omega_c$, so that \eref{eq:indG} holds on this set.
Since a similar argument applies in $G \cap \Omega_3 \setminus \Omega_c$, the proof of  \eref{eq:indG} is complete.

\end{itemize}
Substituting \eref{eq:indOmegaJleq}, \eref{eq:indOmegacleq}, and
\eref{eq:indG}
into \eref{eq:premierLV}, we find
\begin{equs}
LV &\leq  -(1-2\varepsilon-M\varepsilon)e^{\theta H}
-\sum_{j=2,3}\ind_{\Omega_j}(\CstOmegaid - \CstGa\varepsilon-M\varepsilon)p_j^{-2}e^{|p_j|^a + \frac
\theta 2  p_j^2}\\
&\qquad -\ind_{\Omega_c}(\Y\CstOmegacd - \CstGa)e^{\frac \theta 2 (p_2^2 + p_3^2)}
+c\ind_K~.
\end{equs}
Since the constants $C_i$ do not depend on $\varepsilon$ and $M$,
we can make the three parentheses $(1-2\varepsilon-M\varepsilon)$,
$(\CstOmegaid - \CstGa\varepsilon-M\varepsilon)$ and $(\Y\CstOmegacd - \CstGa)$
positive by choosing $M$ large enough and then $\varepsilon$ small enough.
Using again  \eref{e:ptij} and \eref{e:bornesVcOmegac}, we finally obtain
\begin{equs}[eq:Lvavecdenom]
LV & \leq  -ce^{\theta H} 
-c\sum_{j=2,3}\ind_{\Omega_j}\frac{V_j}{p_j^2}-c\ind_{\Omega_c}\frac{V_c}{p_2^2
+p_3^2}+c\ind_K~.
\end{equs}
We now show that this implies \eref{e:ineqLvn}. Observe that since $V \geq
e^{\theta
H}$, we have $\log V \geq \theta H$, and therefore, for $j=2,3$,
$$
p_j^2 \leq p_2^2 + p_3^2 \leq cH+c \leq c\log V + c \leq c(\log V + 2)~.
$$
Since also $-e^{\theta H} \leq -e^{\theta H}/(2+\log V)$, we obtain by
\eref{eq:Lvavecdenom} that
\begin{equs}
LV & \leq  -c\frac{e^{\theta H} 
+\ind_{\Omega_2}{V_2}+\ind_{\Omega_3}{V_3}+\ind_{\Omega_c}{V_c}}{2+\log
V}+c\ind_K = -\frac {c V}{2+\log(V)} + c\ind_K~,
\end{equs}
which, by the definition \eref{e:defphi} of $\phi$, proves \eref{e:ineqLvn}.
\myqed\end{proof}

\section{Proof of \tref{thm:mainthm}}\label{s:controltechnical}

Now that we have a Lyapunov function (\tref{prop:Lyapunov}), we
can prove \tref{thm:mainthm} in the spirit of \cite{cuneo_nonequilibrium_2014}.
In addition to \tref{prop:Lyapunov}, we need a few other ingredients.

We first use the result of \cite{douc_subgeometric_2009} about 
subgeometric ergodicity. We state it here
in a simplified form. For a definition of 
``irreducible skeleton'' and ``petite set'', see the introduction of
\cite{douc_subgeometric_2009} or \cite[Section 2]{cuneo_nonequilibrium_2014}.

\begin{theorem}[Douc-Fort-Guillin (2009)]\label{th:douc}
Assume that a skeleton of the process \eref{e:SDE} is 
irreducible and let $V:\Omega\rightarrow [1,\infty)$ be a smooth
function with $\lim_{\| p\| \to \infty}V( q,  p) = +\infty$.
If there are a petite set $K$ and a constant $C$ such that $LV \leq C \ind_K -
\phi(V)$
for some differentiable, concave and increasing
function $\phi:[1,\infty)\rightarrow(0,\infty)$, 
then the process admits
a unique invariant measure $\pi$, and for any $z\in [0,1]$, there exists a
constant $C'$ such that for all $t\geq 0$ and all $x\in \Omega$,
\begin{equ}\label{eq:convergencedouc}
\|P^t(x, \argcdot) - \pi \|_{(\phi \circ V)^{z}} \leq	g(t) C' V(x)~,
\end{equ}
where $g(t)=(\phi\circ H_\phi^{-1}(t))^{z-1}$,
with $H_\phi(u)=\int_1^u \frac{\dd s}{\phi(s)}$.
\end{theorem}
\begin{proof}
This is a combination of \cite[Theorems 3.2 and 3.4]{douc_subgeometric_2009}
for the following ``inverse Young's functions'' (in the language of \cite{douc_subgeometric_2009}):
 $\Psi_1(s) \propto s^{z-1}$ 
and $\Psi_2(s) \propto s^{z}$.
\myqed\end{proof}

\tref{prop:Lyapunov} provides most of the input to \tref{th:douc}, 
but we still need to check that there is an irreducible skeleton, and
that the set $K$ of \tref{prop:Lyapunov} is petite. To this end, we introduce,
as in  \cite{cuneo_nonequilibrium_2014},
\begin{proposition}\label{prop:controlstart}
The following holds.
\begin{myenum}
	\item[(i)] The transition probabilities $P^t(x, \dd y)$ have a
	$\CC ^\infty((0, \infty)\times \Omega\times \Omega)$ density
	$p_t(x, y)$ and the process is strong Feller.
	\item[(ii)] The time-1 skeleton chain $(x_{n})_{n=0, 1, 2, \cdots}$ admits
	the Lebesgue measure on $(\Omega, \mathcal B)$ as
	a maximal irreducibility measure.
	\item[(iii)] All compact subsets of $\Omega$ are petite.
\end{myenum}
\end{proposition} 
\begin{proof}
(i) follows from H\"ormander's condition. The proof that H\"ormander's
condition holds, which relies on \aref{as:assumptioncoupling},
is very similar to that of Lemma 5.3 of
\cite{cuneo_nonequilibrium_2014} and is left to the reader.
The proof of (ii) is exactly as in Lemma 5.6 of
\cite{cuneo_nonequilibrium_2014}, and (iii) follows from
(i), (ii), and Proposition~6.2.8 of \cite{meyn_markov_2009}.
\myqed\end{proof}

We can now finally give the
\begin{proof}[Proof of \tref{thm:mainthm}]
Let $0 \leq \theta_1 < \min(1/T_1, 1/T_4)$ and $\theta_2 > \theta_1$. Choose
now
$\theta \in (\theta_1, \theta_2)$ such that $\theta< \min(1/T_1, 1/T_4)$. By
\tref{prop:Lyapunov}, we have a Lyapunov function
$1+e^{\theta H}\leq V\leq c(e^{|p_2|^a}+e^{|p_3|^a})e^{\theta H}$ with
$a\in (0,1)$ such that $LV \leq c \ind_K - \phi(V)$, where $\phi(s) =
{c_3\,s}/(2+\log(s))$,
and where $K$ is a compact (and therefore petite) set.  Let now $z\in (0,1)$ be such that
$z\theta > \theta_1$. By \tref{th:douc}, we obtain the existence of
a unique invariant measure $\pi$ such that
\begin{equ}[eq:convergencespecifique]
\|P^t(x, \argcdot) - \pi \|_{(\phi \circ V)^{z}} \leq c e^{-\lambda t^{1/2}}
V(x)~,
\end{equ}
where we have used that with the notation of \tref{th:douc},
\begin{equ}[eq:convrate]
g(t) = (\phi\circ H_\phi^{-1}(t))^{z-1} \leq c e^{-\lambda t^{1/2}}
\end{equ}
for some $\lambda > 0$. Indeed, $H_\phi(u)= \frac
1{c_3}\int_1^u \frac {2+\log s}s \dd s =
 \frac {1}{2c_3} (\log u)^2+\frac 2{c_3} \log u$, so that
$H_\phi^{-1}(t) = \exp({(2c_3t +4)^{1/2}-2})$ and 
$(\phi\circ H_\phi^{-1}(t)) =(2c_3t +4)^{-1/2} \exp((2c_3t +4)^{1/2}-2) \geq
ce^{ct^{1/2}}$, 
which implies \eref{eq:convrate}.

Then, \eref{eq:stretchedexpconv} follows from \eref{eq:convergencespecifique}
and the following two observations. First, we have
$V \leq ce^{\theta_2H}$ since $\theta< \theta_2$.
Secondly, by our choice of $z$, we have $e^{\theta_1H} \leq c(\phi \circ V)^z$, 
so that $\|P^t(x, \argcdot) - \pi\|_{e^{\theta_1H}} \leq c\, \|P^t(x, \argcdot)
- \pi\|_{(\phi \circ V)^{z}}$.

Thus, we have proved (iii). Since (i) and the
smoothness assertion in (ii) follow
from \pref{prop:controlstart}, the proof is complete.
\myqed\end{proof}

\begin{remark}\label{r:conjecture}
It would of course be desirable to generalize \tref{thm:mainthm}
to longer chains of rotors.
The proof of \pref{prop:controlstart} carries on unchanged to chains
of arbitrary length. Therefore, in order to prove the existence
of a steady state
and obtain a convergence rate (with \tref{th:douc}), it ``suffices''
to find an appropriate Lyapunov function. We expect the convergence rate
to be limited by the central rotor (if the length of the chain is odd)
or the two central rotors (if the length is even). Preliminary studies
indicate that for chains of length $n$, a convergence rate 
$ \exp(-c t^{k})$ with $k=1/({2\lceil{n/2}\rceil-2})$ is to be expected.
Obtaining such a result raises some major technical difficulties. First, the
averaging procedure has to be carried to much higher orders,
which quickly becomes intractable
if we proceed explicitly, as we do here. Moreover, the number of regimes to
consider grows very rapidly with $n$. And finally, some generalization
of \pref{prop:rbarref} to more general (nonlinear) decoupled systems will
be needed, with the difficulties mentioned in \rref{rem:paspinningettau}.
We are trying to solve these issues by developing a inductive method which 
requires fewer explicit calculations, but much work remains to be done.
\end{remark}

\appendix

\section{Resonances in the deterministic case}\label{s:resonances}

In \sref{sec:23fast}, two resonant terms appeared, namely ${ {p_1W_{\L}
w_{\C}}/{p_2^{2}}}$
and $-{{p_4W_{\R}w_{\C}}/{p_3^{2}}}$.
These terms have a physical meaning.
We start with the case where
$W_\I(s) = -\cos( s )$, $\I=\L,\C,\R$. Then,
\begin{equs}[eq:decompositionWlWc]
W_\L w_\C=-\cos(q_2-q_1) \sin(q_3-q_2)=  \frac {\sin(q_1- q_3)}2 + \frac{\sin(2q_2-q_3 -
q_1)}2~.
\end{equs}
Consider now the regime where most of the energy is concentrated at sites 2 and
3.
In the approximate dynamics \eref{eq:fastslowtwofast}, we see that
$\sin(q_1- q_3)$ oscillates  with frequency $p_3/2\pi $ and mean zero,
while $\sin(2q_2-q_3 - q_1)$ oscillates with frequency $(2p_2-p_3)/2\pi $.
When $p_3 = 2p_2$, the second term does not oscillate.

\begin{figure}[ht]
\centering
\includegraphics[width=0.95\linewidth, trim = 2mm 2mm -8mm 0, clip]{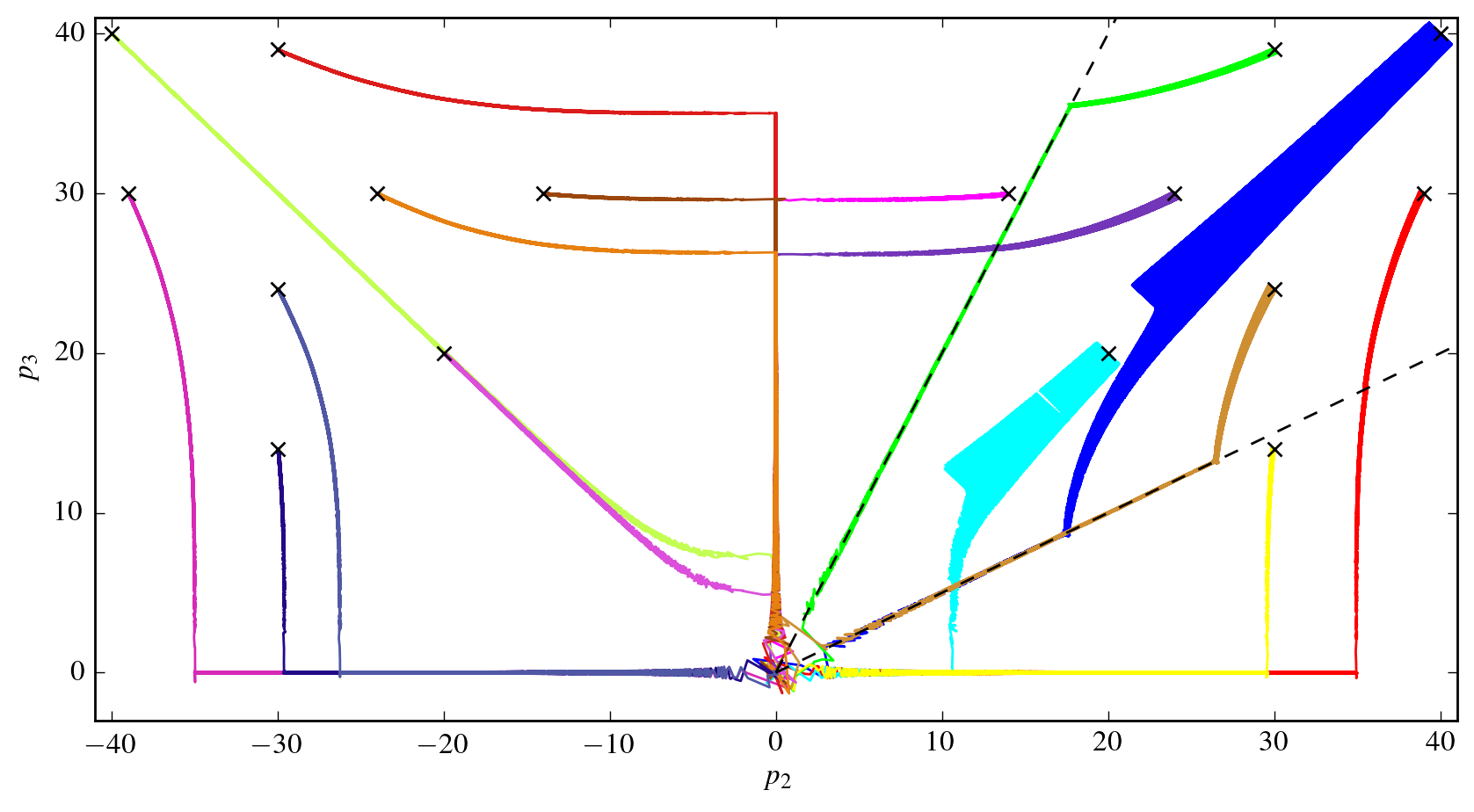}
\caption{Projection of a few orbits on the $p_2p_3$-plane, with
$W_\L = W_\C = W_\R  = -\cos$, $\gamma_1=\gamma_4= 1$, $T_1 = T_4= 0$.
The resonances are depicted as dashed
lines.}\label{fig:n111}
\end{figure}

In \fref{fig:n111}, we represent some trajectories projected onto the
$p_2p_3$-plane in the deterministic case (\ie $T_1 = T_4 = 0$). 
We observe that some trajectories are ``trapped'' by the line
$p_3/p_2 = 2$, while some others just cross it. By symmetry,
the same happens when $p_3/p_2 = 1/2$ because of the term
$-{{p_4W_{\R}w_{\C}}/{p_3^{2}}}$.
This phenomenon does not occur
when the same conditions are used with positive temperatures (see
\fref{fig:traj_avg}). A finer analysis (not detailed here) shows
that in the resonant regime $p_3/p_2 = 2$, a net momentum flux from $p_3$
to $p_2$ appears, and similarly
for $p_3/p_2 = 1/2$ with a flux from $p_2$ to $p_3$. These fluxes
stabilize the resonant regimes.

If we take
$
W_\I(s) = -\cos(n_\I s )
$
for some $n_\I \in \mathbb Z_*$, $\I=\L,\C,\R$, we find by 
a decomposition similar to \eref{eq:decompositionWlWc} some resonances
at
$$
\frac {p_3}{p_2} \in \left\{\frac{n_\C+n_\L}{n_\C} , \frac{n_\C-n_\L}{n_\C},
\frac{n_\C}{n_\C+n_\R}, \frac{n_\C}{n_\C-n_\R} \right\}~.
$$ 
(If some of these values are $0$ or $\infty$, we exclude them since
our approximation is reasonable when both $|p_2|$ and $|p_3|$ are
very large.)
For example, if we choose  $(n_\L, n_\C, n_\R) = (3,1,3)$, we obtain
the ratios $p_3/p_2 = 4, 1/4, -2, -1/2$, which we indeed
observe in \fref{fig:n313}.

\begin{figure}[ht!]
\centering
\includegraphics[width=0.95\linewidth, trim = 2mm 2mm -8mm 0, clip]{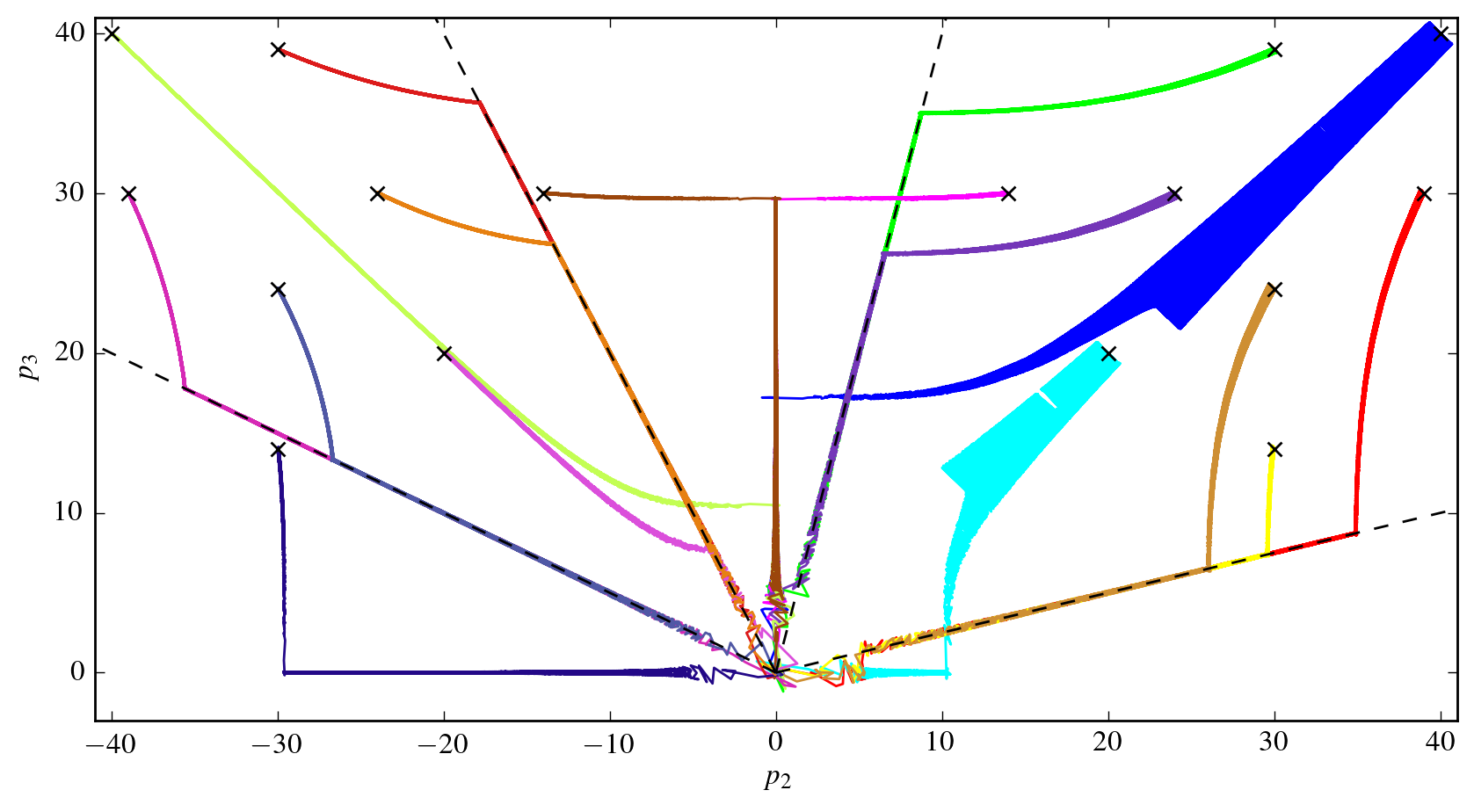}
\caption{Projection of a few orbits on the $p_2p_3$-plane, with
$\gamma_1=\gamma_4= 1$, $T_1 = T_4= 0$, 
and $(n_\L, n_\C, n_\R) = (3,1,3)$. The resonances are depicted as dashed
lines.}\label{fig:n313}
\end{figure}

Of course, a similar analysis applies to more general interaction potentials
by taking their Fourier series and treating the (products of) modes separately.

\subsection*{Acknowledgments} This work has been supported by the
ERC Advanced Grant ``Bridges'' 290843. We are very thankful
to P.~Collet, A.~Dymov, M.~Hairer, V.~Jak\v{s}i\'{c}, C.-A.~Pillet and A.~Shirikyan for 
very helpful discussions, and to C.~Poquet for many valuable suggestions
about the present paper.

\bibliography{refs}

\end{document}